\documentclass{article}
\usepackage{graphicx} 
\usepackage{subfig}
\usepackage{amsfonts}
\usepackage{amsmath}
\usepackage{amssymb}
\usepackage{url}
\usepackage{fancyhdr}
\usepackage{indentfirst}
\usepackage{enumerate}
\usepackage{longtable}
\usepackage{varwidth}
\usepackage{float} 
\usepackage{pifont}
\usepackage{tikz}
\usepackage{tikz-cd}
\usepackage{etoolbox}
\newcommand{\circled}[2][]{\tikz[baseline=(char.base)]
    {\node[shape = circle, draw, inner sep = 1pt]
    (char) {\phantom{\ifblank{#1}{#2}{#1}}};%
    \node at (char.center) {\makebox[0pt][c]{#2}};}}
\robustify{\circled}

\usepackage{dsfont}
\usepackage[colorlinks=true,citecolor=blue]{hyperref}
\usepackage{amsthm}
\usepackage{color}
\usepackage{natbib}
\usepackage{comment}
\usepackage{capt-of}
\usepackage{multirow}

\def\d{\mathrm{d}}

\newcommand{\X}{\mathcal {X}}

\newcommand{\E}{\mathbb{E}}

\newcommand{\F}{\mathcal{F}}

\newcommand{\R}{\mathbb{R}}
\newcommand{\N}{\mathbb{N}}
\newcommand{\p}{\mathbb{P}}

\newcommand{\id}{\mathds{1}}

\renewcommand{\(}{\left(}
\renewcommand{\)}{\right)}
\renewcommand{\[}{\left[}
\renewcommand{\]}{\right]}

\renewcommand{\le}{\leqslant}
\renewcommand{\geq}{\geqslant}
\renewcommand{\leq}{\leqslant}
\renewcommand{\epsilon}{\varepsilon}

\renewcommand{\cdots}{\dots}

\theoremstyle{plain}
\newtheorem{theorem}{Theorem}
\newtheorem{corollary}{Corollary}
\newtheorem{lemma}{Lemma}
\newtheorem{proposition}{Proposition}
\theoremstyle{definition}
\newtheorem{definition}{Definition}
\newtheorem{example}{Example}

\newtheorem{assumption}{Assumption}
\theoremstyle{remark}
\newtheorem{remark}{Remark}

\newcommand{\cet}{\begin{center}}
	\newcommand{\ecet}{\end{center}}

\usepackage{refcount}
\setrefcountdefault{1}

\newcounter{saveexample}

\allowdisplaybreaks[4]
\usepackage{setspace}

\title{
Preference-fitting Framework: Elicited Utility Function and PHARA Approximation}
\author{
    Rui Dai\thanks{\footnotesize Department of Mathematical Sciences, Tsinghua University, China. Email: \texttt{dair25@mails.tsinghua.edu.cn}}
    \and
    Zongxia Liang\thanks{\footnotesize Department of Mathematical Sciences, Tsinghua University, China. Email: \texttt{liangzongxia@tsinghua.edu.cn}}
	\and 		
    Yang Liu\thanks{\footnotesize Corresponding author.  
		School of Science and Engineering, The Chinese University of Hong Kong (Shenzhen), 
		China. Email: \texttt{yangliu16@cuhk.edu.cn}}	
}

\begin{document}
\date{}

\maketitle

\begin{abstract}	
	

The utility function plays a core role in portfolio selection, but its specific form is typically hard to elicit. We propose a definition of the elicited utility function and develop a preference-fitting method to obtain it. Basically, we use intuitive probability-wealth pairs to derive a fitted terminal wealth, a fitted portfolio and a fitted utility function, which converge to the optimal terminal wealth, the optimal portfolio and the elicited utility function of the investor, respectively. Specifically, we first establish a bijection between the utility functions and the terminal wealth functions, based on which we construct the fitted terminal wealth, and then obtain the fitted portfolio and the fitted utility function through the martingale-duality method. Next, we develop a piecewise hyperbolic absolute risk aversion (abbr. PHARA) utility approximation method, and verify the convergences in various senses: almost surely, $L^r$, uniform, etc. We demonstrate two applications of our method: obtaining asymptotically explicit portfolios and handling portfolio selection under Value-at-Risk (abbr. VaR) constraints, thereby illustrating its advantages including intuitiveness, analytical tractability, and ability to circumvent the Lagrange multiplier.

\vspace{0.1in}	
	\noindent \textbf{MSC2020 subject classifications}: Primary: 91G10, 91B16; Secondary: 90C59, 65K10\\

\end{abstract}

\section{Introduction}\label{sec:intro}

In the Black-Scholes model of portfolio selection, 
the classic utility optimization problem (\cite*{M1969}) is given by
\begin{equation}\label{eq:prob_intr}
\begin{aligned}
&\sup_{\boldsymbol{\pi}}\E[U(X_T^{\boldsymbol{\pi}})]
\quad \text{subject to} \quad \d X_t^{\boldsymbol{\pi}} = \left( r X_t^{\boldsymbol{\pi}} + \boldsymbol{\pi}_t^\top (\boldsymbol{\mu} - r \mathbf{1}_m)\right) \d t + \boldsymbol{\pi}_t^\top \boldsymbol{\sigma} \d \mathbf{W}_t,\quad X_0^{\boldsymbol{\pi}} = x_0,
\end{aligned}
\end{equation}
where $T>0$ is a fixed terminal time of investment, $\boldsymbol{\pi}$ is an admissible portfolio representing the investment in risk assets, $\{X_t^{\boldsymbol{\pi}}\}_{0 \leq t \leq T}$ is the corresponding wealth process, $U$ is the utility function, and $x_0$ is the initial wealth value; other parameters will be explained in Section \ref{sec:model}. Building upon this foundational work, a multitude of studies (e.g., \cite*{BS2014}, \cite*{CHEN20191119}, \cite*{HK2018}, \cite*{dong2020optimal}) formulate utility optimization problems in various contexts of finance and insurance. A predominant method employed in these studies to derive the optimal portfolio is the martingale-duality method; 
see \cite*{KLS1987} for a comprehensive proof and implementation. 

In these studies, it is the various utility functions that lead to various optimal portfolios. While most of the literature directly provides specific forms of investors' utility functions, this approach has several drawbacks. First, it is widely recognized that investors struggle to precisely identify the form of their utility functions, as the classic utility is a subjective description rather than an objective requirement. Initially, \cite*{DB1954} suggest the utility function is logarithmic through empirical analysis. \cite*{Tversky1992} propose a novel definition of utility and propose a elicitation method in terms of certainty equivalents. However, both they confine the utility functions to special styles such as S-shaped utility and do not consider the investors' perception of specific market settings. 
Second, directly specifying a utility function often yields counterintuitive optimal portfolios. For instance, the optimal portfolio exposes the investor with S-shaped utility introduced in \cite*{kahneman1979d} to a binary outcome: either bankruptcy or wealth levels exceeding the reference point for risk preference switching. This suggests that the investor is overly aggressive.
Third, explicit expressions for the optimal portfolio can only be derived using particular forms of utility functions through the martingale-duality method. For instance, \cite*{KLSX1991} derive the explicit form of the optimal portfolio for investors with a power utility or logarithmic utility. \cite*{LLMV2024} obtain the explicit optimal portfolio for a general piecewise hyperbolic absolute risk aversion (abbr. PHARA) utility family. This family of utility is widely adopted in various studies, including \cite*{C2000} and \cite*{HK2018}. However, few explicit solutions are attainable for more general cases.

Inspired by \cite*{Tversky1992}, who infer utility functions from intuitive wealth indicators via an optimality criterion, we propose a utility-elicitation framework. Formally, consider an admissible input class $\Xi$, a utility function class $\tilde{\mathcal{A}}$, and a criterion $\delta$, which is a functional defined on $\Xi \times \tilde{\mathcal{A}}$. For any $X^* \in \Xi$, the corresponding utility function $V$ is defined as the element in $\tilde{\mathcal{A}}$ satisfying $\delta_{X^*}(V) = \sup_{X\in \Xi} \delta_X(V)$. For instance, given an investor's expected investment return $X^*$, the utility function $V$ renders $X^*$ optimal among all alternative returns $X$ under a certain criterion $\delta$, (such as maximizing expected utility,) reflecting the investor's preference for $X^*$. For more details, see Definition \ref{def:newutility}. We emphasize that, unlike the traditional optimization problem, this definition directly use satisfactory wealth indicators as input, embedding the optimization property within the criterion. Hence, we still call the satisfactory indicator $X^*$ the optimal indicator, e.g. the optimal terminal wealth. This refined definition of utility holds potential for various applications, such as utility calibration and investment schemes comparison; see Section \ref{sec:model}.

In this paper, we apply the above utility-elicited framework to the Black-Scholes model as the primary case study. Assuming the investor's preferred optimal terminal wealth is $X_T^*$ (representing as a random variable that satisfies some specific conditions; see Section \ref{sec:bijection}), we can elicit the utility function $U$ under the expected utility maximization criterion; see Eq. (\ref{eq:delta}) below. 
This procedure can be viewed as: given the preferred terminal wealth $X_T^*$ in Eq. (\ref{eq:prob_intr}), we infer the corresponding $U$.
The rationale for adopting the expected utility maximization criterion is that investors inherently strive to maximize their investment returns, and investors with different risk preferences (encoded in $U$) choose distinct $X_T^*$. 

However, although the preferred optimal terminal wealth $X_T^*$ is more tractable than the utility function, obtaining closed-form expressions for $X_T^*$ remains challenging. Instead, the investors typically focus on returns under specific market scenarios. These targeted scenarios are expressed as probability-wealth pairs. In our paper, the probability-wealth pair $(p,y)$ is defined as the requirement that an investor obtains a return exceeding $y$ with a probability of $p$. This relationship can be mathematically expressed\footnote{Observe that these pairs adhere to the functional form of the preferred optimal terminal wealth distribution, rendering them a natural proxy for fitting the terminal wealth.} as $\mathbb{P}(X_T^* > y) = p$. 
For any fixed $p$, a larger $y$ entails higher risk. Thus, the selection of pairs requires a trade-off that captures the investor's risk preference. This trade-off reflects the investor's perception of the market model and is explicitly characterized in Section \ref{sec:phara} as the budget bounds, which constitute a necessary condition for the existence of the elicited utility.
The probability-wealth pair is intuitive discrete data, easily accepted by investors, and is widely used in behavioral economics such as \cite*{kahneman1979d} and \cite*{RT2003}. 
Roughly speaking, we obtain the pairs $\{(p_i^n, y_i^n)\}_{0\leq i\leq n}$ and provide a fitted terminal wealth $X_T^n$ that satisfies the requirements of these pairs, which can be achieved by a fitted wealth process $\{X_t^n\}_{0\leq t\leq T}$ driven by a fitted portfolio $\{\boldsymbol{\pi}_t^n\}_{0\leq t\leq T}$. Additionally, this fitted terminal wealth corresponds to a elicited utility function $V_n$ through Eq. (\ref{eq:delta}).  The fitted portfolio has an explicit expression because we formulate $V_n$ as a PHARA utility.
We believe $V_n$ and $\{\boldsymbol{\pi}_t^n\}_{0\leq t\leq T}$ possess practical applicability because they satisfy the requirements across various market scenarios that the investor focuses on. 

As a comparison, the traditional solving procedure for optimization problem is given by
\begin{align*}
\text{A given utility function} & \xrightarrow{\text{Martingale-duality method}} \text{Optimal portfolio}.
\end{align*}
The procedure of our preference-fitting method is given by
\begin{align*}
\text{Probability-wealth pairs} & \xrightarrow{\text{Fitting method}} \text{Fitted terminal wealth}\\ &\left(\xrightarrow{\text{Bijection in Section \ref{sec:bijection}}} \text{Fitted utility function}\right)\\ & \xrightarrow{\text{Martingale-duality method}} \text{Fitted portfolio}.
\end{align*}

In contrast to literature discussing the traditional portfolio selection problem (e.g., \cite*{LIN2017137}, \cite*{dong2020optimal}, \cite*{NS2020}, \cite*{LL2024}), our paper avoids the direct specification of utility functions while still obtaining the explicit solution of portfolio to attain the satisfactory wealth outcome. Additionally, the elicited utility derived in this paper provides a mathematically tighter characterization of the investor's preferences regarding risk; see Section \ref{sec:model}. Compared to behavioral economics research (\cite*{Tversky1992}), we propose a more general elicitation framework, which can be adjusted in different models and can be applied to cases where the investors possess market awareness, e.g., acknowledging the Black-Scholes model.

Moreover, we consider the validity of this method in the mathematical sense. When the fitting accuracy increases, $X_T^n$, $\{X_t^n\}_{0\leq t\leq T}$, $\{\boldsymbol{\pi}_t^n\}_{0\leq t\leq T}$, and $V_n$ converge,\footnote{We will clarify the specific meanings of the various senses of convergence in the following Section \ref{sec:phara}.} and the results are called the optimal terminal wealth $X_T^*$, the optimal wealth process $\{X_t^*\}_{0\leq t\leq T}$, the optimal portfolio $\{\boldsymbol{\pi}_t^*\}_{0\leq t\leq T}$ and the elicited utility function $V$ which serves as a ``real" utility function of the investor, respectively. 

The core of our research consists of three parts: solving the elicitation problem, proposing the preference-fitting method with probability-wealth pairs, and studying the convergences of the fitting procedure. Our main contributions and paper organization are as follows. 

First, in Section \ref{sec:model}, we provide a new definition of the elicited utility function, which can be applied to different models. Taking intuitive indicators as input, the elicited utility function reflects the investor's risk preference under a certain criterion.
The Black-Scholes model serves as our illustrative example. In Section \ref{sec:bijection}, we establish a bijection (\ref{corr}) between all the utility functions and all the terminal wealth functions (see Theorem \ref{th:TL}). Based on the bijection and the martingale-duality method, we solve the elicitation problem under the expected utility maximization criterion Eq. (\ref{eq:delta}). The existence problem and the characterization of equivalence for the elicited utility function\footnote{The elicited utility is not unique for an input $X_T^*$ and all these utility functions form a affine equivalence class $[V]_{\text{aff}}$, i.e. $\{V=aU+b|a>0, b\in \R\}$, where $U$ is a representative element. More importantly, all the solutions corresponding to $X_T^*$ share the same absolute risk aversion function and the same relative risk aversion function, which also determine the optimal portfolios in traditional optimization problems. The rationality of eliciting a representative element will be further discussed in Sections \ref{sec:model} and \ref{sec:bijection}.} are completely resolved. 
\begin{align}\label{corr}
\text{Utility functions} & \longleftrightarrow {\text{Right-hand derivatives}} \longleftrightarrow \text{Terminal wealth functions}.
\end{align}

Second, in Section \ref{sec:fittingmethod}, we propose a preference-fitting method and construct the fitted terminal wealth $X_T^n$, the fitted wealth process $\{X_t^n\}_{0\leq t\leq T}$, the fitted portfolio $\{\boldsymbol{\pi}_t^n\}_{0\leq t\leq T}$ and the fitted utility function $V_n$ using a finite number of probability-wealth pairs $\{(p_i^n,y_i^n)\}_{0\leq i\leq n}$. The existence of these fitting elements can be seen as a direct application of the bijection discussed in Section \ref{sec:bijection}.

Third, in Section \ref{sec:phara}, we exhibit the PHARA approximation approach and verify the effectiveness of the fitting method. We obtain the convergences (e.g. almost surely, $L(\Omega)$, $L(\Omega\times [0,T])$, and uniform convergence) of the fitted wealth process and the fitted optimal portfolio. This indicates that the fitting method performs well both at fixed times and globally, and the convergence exhibits some degree of consistency. We show that the convergence rates for the preference-fitting method are $O(1/n)$ when the terminal wealth function is $C^1$. Moreover, as a byproduct, the PHARA approximation provides an asymptotic explicit form of the optimal portfolio in traditional theory; see Section \ref{sec:ex_hyper}.

Fourth, in Section \ref{sec:VaR}, we demonstrate an economic application and explore the impact of Value-at-Risk (VaR) constraints. Compared with optimization problem introduced in \cite{BS2001}, we analyze how these VaR constraints influence the utility function. We find that introducing VaR constraints is similar to incorporating linear components in the fitting method. Furthermore, compared with classic strategies with VaR constraints, our method is intuitive, providing an explicit optimal portfolio and eliminating the need for discussing Lagrange multipliers.


\section{Model Setting}\label{sec:model}
We use the Black-Scholes framework to model the financial market. The market is assumed to contain $m$ risky assets (stocks). 
The price process   $\{S_{i,t}\}_{0 \leq t \leq T}$   of the $i$-th asset is driven by an $m$-dimensional standard Brownian motion $\{\mathbf{W}_t\}_{0\leq t  \leq T} = \{(W_{1,t}, \cdots, W_{m,t})^\top\}_{0\leq t \leq T}$ on a probability space $\( \Omega, \mathcal{F}, \{\mathcal{F}_t\}_{0 \leq t \leq T}, \p \)$, where  $\{\mathcal{F}_t\}_{0\leq t\leq T}$  denotes the augmentation of the filtration generated by the Brownian motion
$\{\mathbf{W}_t\}_{0 \leq s \leq T}$. Specifically,    $\{S_{i,t}\}_{0 \leq t \leq T}$    follows the following stochastic differential equation ( abbr. SDE):
\begin{equation}
\d S_{i, t} = S_{i, t}\left(\mu_{i} \d t + \sum_{j=1}^m \sigma_{i,j} \d W_{j,t}\right), \quad 0 \leq t \leq T,\; i = 1,\cdots, m,
\end{equation}
where  $\boldsymbol{\mu}:=\(\mu_1, \cdots, \mu_m  \)^\top $ is the vector of expected return rates  on the risky assets  and   $\boldsymbol{\sigma} = (\sigma_{i, j})_{1 \leq i, j \leq m}$ represents the volatility matrix of the market.

Additionally, the financial market includes a risk-free asset $\{S_{0,t}\}_{0\leq t\leq T}$ evolving according to $\d S_{0,t} = r S_{0,t} \d t, \; 0 \leq t \leq T$, where $r$ is the risk-free rate.
We assume that the volatility matrix  $\boldsymbol{\sigma} = (\sigma_{i, j})_{1 \leq i, j \leq m}$ is positive definite, ensuring no-arbitrage conditions within the market model. The invertibility of $\boldsymbol{\sigma}$ and its transpose, along with the positive definiteness of $\boldsymbol{\sigma} \boldsymbol{\sigma}^\top$, also imply market completeness. As such, every contingent claim can be perfectly hedged. 
The process of the asset value $X^{\boldsymbol{\pi}}=\{X^{\boldsymbol{\pi}}_t\}_{0\leq t \leq T}$ is uniquely determined by a strategy process (the wealth allocated on the risky asset) $\boldsymbol{\pi}=\{\boldsymbol{\pi}_t\}_{0\leq t \leq T}$ and an initial value $x_0$. The wealth  process $\{X_t^{\boldsymbol{\pi}}\}_{0 \leq t \leq T}$ under the portfolio $\boldsymbol{\pi}$ is given by the  following SDE:
\begin{equation}\label{eq_notct}
\d X_t^{\boldsymbol{\pi}} = \left( r X_t^{\boldsymbol{\pi}} + \boldsymbol{\pi}_t^\top (\boldsymbol{\mu} - r \mathbf{1}_m)\right) \d t + \boldsymbol{\pi}_t^\top \boldsymbol{\sigma} \d \mathbf{W}_t, \  0 \leq t \leq T, \quad X^{\boldsymbol{\pi}}_0 = x_0,
\end{equation}
where  $\mathbf{1}_m := \(1,\cdots, 1 \)^\top \in \R^m$. When we are not concerned with the strategy, we use the simplified notation $X_t$.
\begin{definition}\label{definition1}
A portfolio $\boldsymbol{\pi} = {\{\boldsymbol{\pi}_t\}}_{0\leq t\leq T }$ is called \textit{admissible} if it satisfies the following conditions:
\begin{enumerate}[i.]
 \item it is progressively measurable with respect to the filtration $\mathcal{F} = \{\mathcal{F}_t\}_{0 \leq t \leq T}$; 
 \item it satisfies $\int_0^T ||\boldsymbol{\pi}_t||_2^2 \d t < \infty$ almost surely;
 \item Eq.\eqref{eq_notct} admits a unique strong solution starting from $X^{\boldsymbol{\pi}}_0 = x_0$;
 \item there exists a constant $C \geq 0$ such that $X^{\boldsymbol{\pi}}_t +C e^{rt} \geq 0$ almost surely for any $t\in[0,T]$.
\end{enumerate}
Denote $\Pi$ as the collection of all the admissible portfolios.
\end{definition}
Condition ii ensures the well-definedness of the stochastic integral, while $Ce^{rt}$ in Condition iv is similar to the liquidation boundary in \cite*{HK2018}, indicating tolerance of bankruptcy.

Throughout this paper, the pricing kernel process $\xi=\{ \xi_t\}_{0\leq t\leq T}$ is defined by
\begin{equation}\label{eq_xi}
    \xi_t := \exp\left\{-\(r+\frac{1}{2}||\boldsymbol{\theta}||_2^2\)t - \boldsymbol{\theta}^\top \mathbf{W}_t\right\}, \quad 0\leq t \leq T,
\end{equation}
where $\boldsymbol{\theta} := \boldsymbol{\sigma}^{-1} (\boldsymbol{\mu} - r \mathbf{1}_m)$ represents the market price of risk. 

\begin{remark}\label{remark:supermartingale}
For any admissible $\boldsymbol{\pi} \in \Pi$, the process $\{\xi_t X_t^{\boldsymbol{\pi}}\}_{0\leq t\leq T}$ is a supermartingale. In fact, using Ito's formula and  Eqs.\eqref{eq_notct}-\eqref{eq_xi}, we have
\begin{equation}\label{eq_martingale}
\d (\xi_tX^{\boldsymbol{\pi}}_t) =\xi_t({\boldsymbol{\pi}}_t^\top \boldsymbol{\sigma}-X^{\boldsymbol{\pi}}_t{\boldsymbol{\theta}}^\top )\d \mathbf{W}_t ,  \  t\leq T, \quad X^{\boldsymbol{\pi}}_0 = x_0.
\end{equation}
Hence, $\{\xi_tX^{\boldsymbol{\pi}}_t\}_{0\leq t\leq T}$ is a continuous local martingale. In addition, we can conclude from Eq. \eqref{eq_xi} that $\left\{e^{rt}\xi_t\right\}_{0\leq t\leq T}$ is a martingale. It follows that $\{\xi_t(X^{\boldsymbol{\pi}}_t+Ce^{rt})\}_{0\leq t\leq T}$ is a non-negative continuous local martingale and thus a supermartingale. Finally, we have that $\{\xi_tX^{\boldsymbol{\pi}}_t\}_{0\leq t \leq T}=\{\xi_t(X^{\boldsymbol{\pi}}_t+Ce^{rt})-Ce^{rt}\xi_t\}_{0\leq t \leq T}$ is a supermartingale. 

Therefore, we ensure that the required supermartingale assumption 
in the martingale-duality method holds (\cite*{KLS1987}). On this basis, we can follow all steps of the martingale-duality method.
\end{remark}
\begin{definition}\label{def:utilityfunction}
The utility function $U$ refers to an increasing, concave and continuous function defined on $[0, \infty)$ with $\lim\limits_{x\to \infty}\frac{U(x)}{x}=0$.
\end{definition}
The monotonicity of the utility function implies that greater wealth yields higher satisfaction for the investor. The condition $\lim\limits_{x\to \infty}\frac{U(x)}{x}=0$ is the classic Inada condition in \cite*{KLSX1991}. 

In the traditional portfolio-selection problem, an investor with a given utility function solves an optimization problem to conduct portfolio selection: 
\begin{equation}\label{prob-hat}
	\begin{aligned}
		& \sup_{\boldsymbol{\pi} \in \Pi} \E[ U(X^{\boldsymbol{\pi}}_T)],
	\end{aligned}
\end{equation}
where $U$ is her utility function. Various literature have derived the optimal portfolio by the martingale-duality method for different utility functions.
\begin{assumption}\label{ass:L2}
We assume $\sup_{\boldsymbol{\pi} \in \Pi} \E[ U(X^{\boldsymbol{\pi}}_T)]<\infty$ and $\E[(\xi_TI(\nu\xi_T))^2]<\infty$ for any $\nu>0$
all through this paper, where $I=(U_+')^{-1}$ is the inverse of the right-derivative of $U$.
\end{assumption}
The right-derivative of $U$ exists due to the concavity. These assumptions are quite weak, which guarantee the conditions of the martingale representation theorem in the procedure of the martingale-duality method. Utility functions adopted in most studies such as \cite*{HK2018} and \cite*{LIN2017137} satisfy this condition.
The conclusion below shows that any optimal solution coincides almost surely with the solution obtained by the martingale-duality method.
\begin{proposition}\label{prop:unique}
Under Assumption \ref{ass:L2}, for utility function $U$ introduced in Definition \ref{def:utilityfunction}, the optimization problem (\ref{eq:prob_intr}) admits a unique optimal portfolio $\boldsymbol{\pi}^*$, and a unique optimal terminal wealth $X_T^*$ in the sense of almost surely. That is, for optimal terminal wealth $X_T^1$ and $X_T^2$, we have $\p(X_T^1=X_T^2)=1$.
\end{proposition}
\begin{proof}
Using the martingale-duality method, we obtain the existence of $\boldsymbol{\pi}^*$ (\cite*{KLSX1991}). Then we prove the uniqueness. Assume $U$ is not constant without loss of generality. Let $X^*_T:
\Omega\to [0,\infty)$ be the optimal terminal wealth. For any $a,b\in [0,\infty)$ such that $U$ is linear (i.e. $U=k_1x+k_2$ for $k_1>0$ and $k_2\in \R$) on $(a,b)$, we prove that $\p(X_T^*\in (a,b))=0$. Assume $\p(X_T^*\in (a,b))>0$. Let $\Omega_1=\{\omega|X_T^*(\omega)\in (a,b)\}$, then $\p(\Omega_1)>0$. Using the intermediate value theorem we deduce that there exists $z\in \R^+$ such that $$a\cdot\E[\xi_T\id_{\Omega_1\cap\{\xi_T> z\}}]+ b\cdot\E[\xi_T\id_{\Omega_1\cap\{\xi_T\leq z\}}]=\E[\xi_TX_T^*\id_{\Omega_1}].$$
Define 
\begin{equation*}
\overline{X}_T(\omega)=\begin{cases}
a, &\omega \in  \Omega_1\cap\{\omega|\xi_T(\omega)> z\},\\
b, &\omega \in  \Omega_1\cap \{\omega|\xi_T(\omega)\leq z\},\\
X_T^*(\omega), &\omega \in \Omega \setminus \Omega_1.
\end{cases}
\end{equation*}
We have $\E[\xi_TX_T^*]=\E[\xi_T\overline{X}_T]$ and thus $\overline{X}_T$ is attainable through an admissible portfolio.
Note 
$\E[(U(\overline{X}_T)-U(X_T^*))(\xi_T-z)]<0$ holds as $U(\overline{X}_T(\omega))\geq U(X_T^*(\omega))$ when $\omega \in\Omega_1\cap\{\xi_T\geq z\}$, $U(\overline{X}_T(\omega))\leq  U(X_T^*(\omega))$ when $\omega\in \Omega_1 \cap \{\xi_T> z\}$, $U(\overline{X}_T(\omega))=U(X_T^*(\omega))$ when $\omega\in \Omega\setminus \Omega_1$, and $\p(\Omega_1\cap\{\overline{X}_T-X_T^*>0\})>0$. Hence, we have $z\cdot(\E[U(\overline{X}_T)]-\E[U(X_T^*)])>k_1\cdot\E[\xi_T(\overline{X}_T-X_T^*)\id_{\Omega_1}]=0$, which contradicts to the optimality of $X_T^*$. Hence, we have $\p(X_T^*\in(a,b))=0$.

For optimal terminal wealth $X_T^1$ and $X_T^2$, we have $\E[U(X_T^1)]=\E[U(X_T^2)]$. Because of the concavity of $U$, we have
$$\E\Big[U\Big(\frac{X_T^1+X_T^2}{2}\Big)\Big]\geq \frac{\E[U(X_T^1)]+\E[U(X_T^2)]}{2}=\E[U(X_T^1)].$$
As $X_T^1$ is optimal, we have the first inequality becomes a equality.
Because $X_T^1$ and $X_T^2$ lies in the strictly concave region of $U$ (with zero probability of lying in the linear part), and $ U(\frac{X_T^1+X_T^2}{2})>\frac{U(X_T^1)+U(X_T^2)}{2}$ if and only if $X_T^1\ne X_T^2$, we conclude $X_T^1=X_T^2$ a.s.. Therefore, we complete the proof.
\end{proof}
However, most studies directly assume a specific functional form for the utility function to solve the optimization problems, which, as discussed in Section \ref{sec:intro}, is not entirely reasonable. 
\cite*{Tversky1992} establish intuitive satisfactory wealth outcomes (certainty equivalence therein) and then infer the utility function via a specific preference criterion (minimal regression error therein), inspiring us to give the following definition.
\begin{definition}[Elicited utility function]\label{def:newutility}
Denoted by $\Xi$ the set of all the wealth outcomes $X:\omega\mapsto [0,\infty)$, by $\tilde{\mathcal{A}}$ the set of utility functions, and by $\delta:\tilde{\mathcal{A}}\times \Xi\to \R$ the elicitation criterion functional. Assuming the optimal terminal wealth of the investor is given by $X^*\in \Xi$, we call $V\in\tilde{\mathcal{A}}$ the elicited utility function of this investor under the criterion $\delta$ if it satisfies
$$\delta_{X^*}(V)=\sup_{X\in \Xi}\delta_{X}(V).$$
\end{definition}
Intuitively, the utility function $V$ renders the outcome $X^*$
optimal among all alternatives $X$, thereby reflecting the investor's preference for $X^*$. The utility is dependent on the appointed $\delta$ and the input $X^*$. Because the criterion embodies the optimality, we still call the preferred outcome $X^*$ the optimal wealth when there is no ambiguity with traditional optimization problems. Within specific models, this elicited utility function admits several applications. For example, in personalized asset allocation, it replaces the rough assumption of specific utility such as uniform constant relative risk aversion (abbr. RRA) utility (\cite*{KS1991}) in traditional models with a quantitative measure of actual investor preferences, enabling utility calibration (e.g., see Remark \ref{remark:beha}). It also permits dynamic adjustments when the environment or investor psychology changes. Moreover, we can employ the elicited utility function to simulate the investor's response to previously untried strategies $\boldsymbol{\pi}$, using the magnitude of $\E[U(X_T^{\boldsymbol{\pi}})]/\E[U(X_T^*)]$ as the evaluation standard. This is a capability particularly relevant when the optimal strategy becomes infeasible or when investors should select diverse investment strategies provided by different firms.

We adopt this elicitation method to Black-Scholes model. We believe the investors typically focus on returns in specific scenarios and are quite sensitive to probability-wealth pairs instead of the utility functions based on the history data, the investment object, and the inner preference. This means they know the satisfactory probability $p$ for obtaining a wealth exceeding a specific level $y$ while they also understand the trade-off: increasing $p$ for a given $y$ sacrifices the opportunity of attaining other wealth levels\footnote{Equivalently, the pair represents the wealth threshold $y$ above which the investor is satisfied for a given probability $p$.} (we will exhibit this mathematically through budget constraints in Section \ref{sec:fittingmethod}). Consequently, investors with different risk preferences prescribe distinct sets of $(p, y)$ pairs. Obviously, the pair $(p, y)$ and the investor's preferred optimal terminal wealth $X_T^*$ satisfy the relation $\mathbb{P}(X_T^* > y) = p$.
The pairs are more intuitive than the distribution of $X_T^*$ because they are discrete data rather than continuous functional expression. Hence, with these pairs, we can fit the optimal terminal wealth. We merely require investors to specify a finite number of pairs to achieve favorable practical outcomes.

Next, we choose the elicitation criterion for the input (fitted) $X_T^*$. When investors weigh investment risks while seeking the highest possible returns inherently, their utility function $U$ is expected to satisfy
\begin{equation}\label{eq:delta}
\E[U(X_T^*)]=\sup_{X_T\in \Xi}\E[U(X_T)],
\end{equation}
where $X_T^*$ is the preferred terminal wealth of the investor which is optimal under the criterion, and $X_T$ represents other alternative terminal wealth distribution. This means that the utility function renders $X_T^*$ the most satisfactory outcome for the investor and thus captures risk preferences. Under this elicitation criterion, we can compare the utility function elicited for investors’ preferred optimal wealth with those directly specified in existing studies to conduct utility calibration.

Recalling Definition \ref{def:newutility}, in the remainder of this paper, $\tilde{\mathcal{A}}$ consists of all the utility functions and $\delta$ is given by Eq. (\ref{eq:delta}). The class $\Xi$ consist of all the terminal wealths that can be attained by an admissible portfolio (we will specify this in the next section). 
Hence, given the $(p,y)$ pairs, what we should do is to fit the optimal wealth $X_T^*$ using the fitted terminal wealth $X_T^n$ and then solve for the elicited function $V_n$ satisfying $\E[V_n(X_T^n)]=\sup_{X_T\in \Xi}\E[V_n(X_T)]$. Moreover, we are particularly concerned with the convergence of the aforementioned method and we obtain an explicit form of the portfolio via the martingale-duality method. In Section \ref{sec:bijection}, we infer the utility function from a given optimal terminal wealth $X_T^*\in \Xi$ and the remaining mentioned content is completed in Sections \ref{sec:fittingmethod} and \ref{sec:phara}.

Problem (\ref{eq:delta}) is solvable if and only if the input optimal terminal wealth $X_T^*$ is a solution of Problem (\ref{eq:prob_intr}) for some specific utility function $U$. Hence, we introduce the following definition of utility-based investor.
\begin{definition}\label{def:utilitybased}
We call an investor a utility-based investor if her terminal wealth can be represented by the optimal solution $X_T^{\boldsymbol{\pi}^*}$ of Problem (\ref{eq:prob_intr}) for a specific utility function $U$.
\end{definition}
For $X_T^*$, the elicited utility function exists if and only if $X_T^*$ is the wealth of a utility-based investor.
An important feature of the terminal wealth for a utility-based investor is decreasing with respect to $\xi_T$, where $\xi_T$ is defined by Eq. (\ref{eq_xi}); see \cite*{LLMV2024}. In fact, under standard conditions in complete market, the unique solution of Problem (\ref{eq:prob_intr}) exists and satisfies $X_T^*=I(\nu^*\xi_T)$, where $I = (U')^{-1}$ a.e. and $\nu^*$ is a unique Lagrange multiplier satisfying $\E\left[ \xi_T I(\nu^* \xi_T) \right] = x_0$. Hence, the optimal wealth is decreasing\footnote{More precisely, according to Proposition \ref{prop:unique}, any two optimal terminal wealths coincide almost surely, so the impact of any exceptional set is negligible. Consequently, we consider only the one that satisfies monotonicity.} with respect to $\xi_T$ due to the monotonicity of $I$. In Section \ref{sec:bijection}, we will demonstrate that if a given terminal wealth $X_T$ is decreasing with respect to $\xi_T$, and satisfies some other conditions, then the investor is qualified as utility-based.
Therefore, for Problem (\ref{eq:delta}), if an input $X_T^*$ satisfies these conditions, we can assert the existence of a corresponding elicited utility. We will explain the meaning of these conditions in Remark \ref{remark:utilitybased}. Regarding uniqueness, we make the following statement.

We point out that for a given $X_T^*$, the elicited utility function $U$ serving as a solution of Problem (\ref{eq:delta}) is not unique.  However, all such utility functions share the same absolute risk aversion (abbr.ARA) function and the same relative risk aversion (abbr. RRA) function almost everywhere (see the next section), and thus, we consider them equivalent\footnote{Our method essentially elicits the investor's ARA and RRA function, which have a one-to-one correspondence with the investor's optimal portfolios. These functions are precisely the key to characterizing intrinsic risk preference.}. Notably, for Problem (\ref{eq:prob_intr}), it is the ARA and RRA function that determine the investor's optimal portfolio, and different utility functions may also lead to the same optimal terminal wealth. Furthermore, two utility functions $U_1$ and $U_2$ are equivalent if and only if there exist $a>0$ and $b\in \R$ such that $U_2=aU_1+b$. Therefore, the utility function we elicit is a representative element and we can easily check the equivalence through the relationship $U_2=aU_1+b$.
\begin{remark}
In the work of \cite*{LLZ2024}, the definition of utility functions is extended to increasing, right-continuous functions (the continuity and concavity are not necessary), including the S-shaped utility presented in \cite*{Tversky1992}, and it is shown that the optimization problem (\ref{eq:prob_intr}) remains solvable. Hence, we can accordingly extend Definition \ref{def:utilityfunction} to this setting. In this case, utility functions $U_1$ and $U_2$ are equivalent if and only if their concave envelopes $U_1^{**}$ and $U_2^{**}$ correspond to the same ARA function almost everywhere, or equivalently, there exists $a>0$ and $b\in \R$ such that $U_2^{**}=aU_1^{**}+b$. Here, the concave envelope is defined as 
\begin{align*}
U^{**}(x):=\inf\{h: \mathcal{D} \to \R \; | \text{ $h$ is a concave function and $h\geq U$}\},
\end{align*}
which is a function that satisfies Definition \ref{def:utilityfunction}.
This is because $\sup\limits_{\boldsymbol{\pi}} \E[U(X^{\boldsymbol{\pi}}_T)]=\sup\limits_{\boldsymbol{\pi}} \E[U^{**}(X^{\boldsymbol{\pi}}_T)]$ holds in traditional portfolio selection problem (\ref{eq:prob_intr}); see e.g., \cite*{LLMV2024}.
\end{remark}
\section{Bijection between Utility Functions and Terminal Wealth Functions}\label{sec:bijection}
The primary objective of this section is to solve Problem (\ref{eq:delta}) for a given $X_T^*$. For Problem (\ref{eq:prob_intr}), giving a utility function $U$, we can obtain the corresponding optimal terminal wealth $X_T^*$ using the martingale-duality method. In our framework, for Problem (\ref{eq:delta}), we seek the inverse of the above solution procedure. To achieve this, we first present a bijection to characterize the features of the terminal wealth for utility-based investors. The primary significance of this bijection is determining whether an investor is utility-based on her terminal wealth and demonstrating how to obtain the utility function $U$ form $X_T^*$. 

In the following definition, we define the class $\mathcal{A}$ to represent all utility functions and define the class $\mathcal{C}$  to represent all terminal wealth functions. The terminal wealth function is the generalized inverse of the right-hand derivative of the utility function. We will further illustrate these definitions in Remark \ref{remark:terminalwealth}. To simplify the discussion, we assume that the utility functions of the investor are strictly increasing without loss of generality. For the function $V$ defined on $(0,\infty)$ with $\lim\limits_{x\to 0^+}V(x)=-\infty$, we extend $V(0)=-\infty$ and say that $V$ has the domain $[0,\infty)$ and is right continuous at $x=0$. We establish a bijection $T_L$ from $\mathcal{A}$ to $\mathcal{C}$, providing the specific correspondence of the mapping, which is presented in Theorem \ref{th:TL}.
\begin{definition}\label{def:fuc-class}
We define three types of function classes as follows: \\
\textbf{Function class $\mathcal{A}$ (Utility functions)}: the set of all functions $V$ with domain $[0, \infty)$, which are concave, right continuous at 0, strictly increasing, and satisfy $V(1) = 0$ and $\lim\limits_{x\to\infty}V'_+(x)=0$, where $V'_+$ is the right-hand derivative of $V$; \\
\textbf{Function class $\mathcal{B}$ (Right-hand derivatives)}: the set of all functions $f_r$ with domain $[0,\infty)$, which are always greater than $0$, decreasing, right continuous, and satisfy $\lim\limits_{x \to \infty} f_r(x)=0$;\\
\textbf{Function class $\mathcal{C}$ (Terminal wealth functions)}: the set of all functions $f_l$ with domain $[0,\infty)$, which are nonnegative, decreasing, left continuous, and satisfy $\lim\limits_{x \to 0^+} f_l(x)=\infty$. 
\end{definition}
\begin{remark}\label{class}
We can directly prove that $V$ is continuous in its domain due to the properties of concavity. Moreover, the condition $V(1)=0$ is reasonable because adding a constant to the utility functions does not affect the form of the optimal wealth and the optimal portfolio. 
The requirement $\lim\limits_{x\to\infty}V'_+(x)=0$ is meaningful as the right-hand derivative (can take values of $\infty$ at $x=0$) of a concave function exists in the interval $[0, \infty)$. 
For the definition of $\mathcal{B}$, in the cases where $\lim\limits_{x\to0^+}f_r(x) =\infty$, we extend $f_r(0)=\infty$. For the definition of $\mathcal{C}$, we extend $f_l(0)=\infty$ as $\lim\limits_{x\to0^+}f_l(x) =\infty$ always happens. The class $\mathcal{B}$ can be seen as the set of right-hand derivative functions formed from the set $\mathcal{A}$. 
\end{remark}


We give the specific mapping relationships in the following definition.
\begin{definition}\label{def:bij}
Define an operator $D_+:\mathcal{A} \to \mathcal{B}$ as $V \mapsto V'_+$. Define an operator $T_0:\mathcal{B} \to\mathcal{A}$ as $f_r \mapsto V_r$, where $V_r(x)=\int_1^{x} f_r(t) \d t$. The operator $T_1:\mathcal{B} \to \mathcal{C}$ is defined as follows: For each function $f_r \in \mathcal{B}$, $T_1(f_r)$ is a function defined on $[0,\infty)$ satisfying $T_1(f_r;y)=\sup\limits_{x\geq 0}\{x:f_r(x)\geq y\}$. The operator  $T_2:\mathcal{C} \to \mathcal{B}$ is defined as follows: For each function $f_l \in \mathcal{C}$, $T_2(f_l)$ is a function defined on $[0,\infty)$ satisfying $T_2(f_l;y)=\sup\limits_{x\geq 0}\{x:f_l(x)> y\}$. The operator $T_L:\mathcal{A}\to \mathcal{C}$ is defined as follows: For each $V\in \mathcal{A}$, $T_L$ is a function defined on $[0,\infty)$ satisfying $T_L(V;y) = \sup\{{\mathrm{argmax}}_{x\geq 0} (V(x)-xy)\} $. By convention, we designate $\sup\{\varnothing\}=0$. The mapping relationships are illustrated as follows:
\begin{figure}[ht]
    \centering
    \begin{minipage}{0.45\textwidth}
        \centering
        \begin{tikzcd}[row sep=2cm, column sep=2cm]
        \mathcal{A} \arrow[r, "D_+"{above}]    
        & \mathcal{B} \arrow[r, "T_1"{above}]  
        & \mathcal{C} \arrow[ll, "T_L^{-1}"{below}, bend left=30]  
    \end{tikzcd}
    \end{minipage}%
    \hfill
    \begin{minipage}{0.45\textwidth}
        \centering
        \begin{tikzcd}[row sep=2cm, column sep=2cm]
        \mathcal{C} \arrow[r, "T_2"{above}]    
        & \mathcal{B} \arrow[r, "T_0"{above}]  
        & \mathcal{A} \arrow[ll, "T_L"{below}, bend left=30]  
    \end{tikzcd}
    \end{minipage}
\end{figure}
\end{definition}
\begin{remark}\label{remark:terminalwealth}
Based on the martingale-duality method, we have the optimal terminal wealth $X_T^*$ of Problem (\ref{eq:prob_intr}) satisfying $X_T^*(\omega) = T_{L}(U;\nu^*\xi_T(\omega))$, where the corresponding utility function is $U$, and $\nu^*$ is the Lagrange multiplier satisfying the budget constraint
\begin{equation*}
\E\left[\xi_TT_L(U;\nu^*\xi_T)\right]=\E[\xi_Tf_l(\nu^*\xi_T)]=x_0,
\end{equation*} 
which guarantees the condition of the martingale representation theorem; see, e.g.,  \cite*{LL2024}, \cite*{LIN2017137}, \cite*{HK2018} and reference therein. The class $\mathcal{A}$ extracts the common properties of the utility functions while the class $\mathcal{C}$ extracts those of all the functions with the form $f_l=T_L(U)$.

We emphasize that the term ``terminal wealth function" differs from ``terminal wealth". In fact, for a utility-based investor, the optimal terminal wealth $X_T^*$ can be derived by the martingale-duality method using the terminal wealth function $f_l$, the pricing kernel $\xi_T$ and the Lagrange multiplier $\nu^*$, and specifically, $X_T^*=f_l(\nu^*\xi_T)$.
\end{remark}
\begin{lemma}\label{inverse}
The mappings $D_+$ and $T_1$ are bijections from $\mathcal{A}$ to $\mathcal{B}$ and from $\mathcal{B}$ to $\mathcal{C}$, respectively. Moreover, $D_+^{-1}=T_0$ and $T_1^{-1} = T_2$.
\end{lemma}
\begin{proof}
See Appendix \ref{appendix_A}.
\end{proof}
\begin{theorem}\label{th:TL}
$T_L$ is a bijection from $\mathcal{A}$ to $\mathcal{C}$. Moreover, $T_L = T_1 \circ D_+$ and $T_L^{-1}= T_0 \circ T_2$.
\end{theorem}
\begin{proof}
We first prove $T_L = T_1 \circ D_+$, i.e., the equality
$$\sup\{{\mathrm{argmax}}_{x} (V(x)-xy)\} =\sup\{x:D_+V(x)\geq y\}$$ holds for any $y\in [0,\infty)$. Particularly, when $y=0$, we see that both sides of the equation are $\infty$. For $y>0$, we define $$x_+ \triangleq \sup\{x:D_+V(x)\geq y\},$$ and then have   $x_+<\infty$ due to $\lim\limits_{x\to \infty}D_+V(x)=0$. Obviously, $D_+V(x)< y$ holds for any $x>x_+$. Hence, $D_+V(x_+)\leq y$ as $D_+(V)$ is right continuous. In addition, we easily conclude that $V'_-(x^+)\geq y$ because  for any $x<x_+$, it holds that $V'_-(x)\geq V'_+(x)\geq y$ and $V'_-$ is left continuous. The property of Legendre's transform yields 
$$\mathrm{argmax}_x (V(x)-xy)=\{x:V'_+(x)\leq y \leq V'_-(x)\}.$$ Then  $x_+ \in \{x:V'_+(x)\leq y \leq V'_-(x)\} $. Because $V'_-(x)<y$ holds for any $x>x_+$, we conclude  that $x_+ = \sup\{\mathrm{argmax}_x (V(x)-xy)\} =\max\{\mathrm{argmax}_x (V(x)-xy)\}$. Thus, $T_L = T_1 \circ D_+$. Finally, using Lemma \ref{inverse}, we have that $T_L$ is a bijection and $T_L^{-1}= T_0 \circ T_2$.
\end{proof}
\begin{remark}\label{remark:X_T^*}
On one hand, the terminal wealth function for an investor with utility function $V$ is $T_L(V)\in\mathcal{C}$. On the other hand, if an investor's terminal wealth function belongs to the class $\mathcal{C}$, we can derive the corresponding utility function in $\mathcal{A}$ via the inverse mapping of $T_L$.
According to Remark \ref{remark:terminalwealth}, there exists a decreasing function $\mathcal{X}_T^*:\xi \mapsto f_l(\nu^*\xi)$ for any $\xi\geq 0$ such that $X_T^*=\mathcal{X}_T^*(\xi_T)$, revealing that the optimal terminal wealth for Problem (\ref{eq:prob_intr}) can be regarded as a function $\mathcal{X}$ with respect to the pricing kernel $\xi_T$ and $\mathcal{X}$ belongs to the class $\mathcal{C}$ as the single proportional linear transformation does not change the properties of functions in class $\mathcal{C}$.
\end{remark}

To ensure the existence of the elicited utility function, we must properly define $\Xi$ for Problem (\ref{eq:delta}), as stated in Section \ref{sec:model}.
We derive the sufficient and necessary conditions\footnote{Existence results for optimization problems in incomplete markets are abundant, but explicit expressions for the price kernel remain scarce. This paper restricts attention to complete markets, where the terminal wealth and portfolio of a utility-based investor admit an explicit characterization.} under which a function $\mathcal{X}_T$ constitutes an optimal terminal wealth in the form $X_T^*=\mathcal{X}_T(\xi_T)$ for some utility-based investor with an initial asset $x_0$: first, $\mathcal{X}_T$ belongs to the class $\mathcal{C}$; and second, the budget constraint $\E[\xi_T\mathcal{X}_T(\xi_T)]=x_0$ is satisfied. Thus, we define $\Xi$ as the set of all the $\X$ satisfying these two conditions.
For the sufficiency, based on Remarks \ref{remark:terminalwealth} and \ref{remark:X_T^*}, if these conditions are fulfilled, according to the bijection established in Theorem \ref{th:TL} and the martingale-duality method, we have that the optimal terminal wealth of the investor with the utility function $V:x \mapsto T_L^{-1}(\mathcal{X}_T)$ is precisely $X_T^*$ and the corresponding Lagrange multiplier is equal 
to 1. This implies that by specifying a terminal wealth that satisfies the aforementioned two conditions, we can conclude that the investor is utility-based, and thus the sufficiency holds. For the necessity, we have the unique optimal terminal wealth is given by $X_T^*=T_L(U;\nu^*\xi_T)$ according to Proposition \ref{prop:unique} and Remark \ref{remark:X_T^*}, which satisfies the above two conditions. Hence, the necessity is obviously holds.

As outlined in Section \ref{sec:model}, for $\X$ satisfying the above two conditions, the elicited utility function corresponding to the optimal terminal wealth $X_T^*=\X(\xi_T)$ is not unique and we can prove all these utility functions form the affine equivalence class $[V]_{\text{aff}}=\{\nu^{-1} \cdot T_L^{-1}(\mathcal{X})+a \mid \nu \in \mathbb{R}^+,a\in \R\}$.
In fact, according to Remark \ref{remark:X_T^*}, we assume $X_T^*=f_{l,1}(\xi)=f_{l,2}(\nu\xi)$ for any $\xi>0$, where $\nu>0$ satisfies $\E[\xi_Tf_{l,1}(\xi_T)]=\E[\xi_Tf_{l,2}(\nu\xi_T)]=x_0$. Hence, we have $f_{l,1}(\xi)=f_{l,2}(\nu\xi)$ for any $\xi>0$. Then, using the properties of $T_0$ and $T_2$, we obtain $V_1=T_L^{-1}(f_{l,1})=\nu T_L^{-1}(f_{l,2})-b$ for some $b\in\R$ and the result follows.
Without loss of generality, we always take $\nu=1$ in the following, and thus the optimal terminal wealth is given by $X_T=\X(\xi_T)=f_l(\xi_T)$. Therefore, we can regard $\mathcal{C}$ as the set of all optimal terminal wealths.
\begin{remark}\label{remark:utilitybased}
The investor is utility-based if and only if her terminal wealth satisfies the above two conditions.
The first condition is natural: it states that investors expect higher returns when market conditions are better. The second condition is endogenous to the Black-Scholes model and is independent of the martingale-duality method; it implies that the investor's objective is to maximize wealth as much as possible.
Specifically, in the Black-Scholes model, the vector $\boldsymbol{\theta}$ represents the market price of risk, while a larger realization of $\mathbf{W}_T$ corresponds to better performance of the risky asset. Noting that $\xi_T$ is decreasing in $\boldsymbol{\theta}^\top \mathbf{W}_T$, we consider that a larger $\xi_T$ reflects worse market conditions. The second condition reflects a trade-off: as the pricing-kernel-weighted expectation of terminal wealth is fixed, pursuing a higher wealth level at a fixed probability forces the investor to forgo the opportunity of attaining higher wealth levels at other probability levels. We argue that the second condition is reasonable. For any admissible portfolio $\boldsymbol{\pi}$, we have $\E[\xi_TX^{\boldsymbol{\pi}}_T]\leq x_0$ (this result is inherent in the Black-Scholes model), as shown in Remark \ref{remark:supermartingale}. When the above equality does not hold, the investor fails to pursue greater wealth sufficiently; in other words, the investor could have increased the expected terminal wealth at no risk cost. Hence, this condition broadly characterizes the requirement that a profit-maximizing investor who measures risk while inherently pursuing return maximization should satisfy.
\end{remark}
Then we consider Problem (\ref{eq:delta}). Define $\Xi$ as all the terminal wealth $X_T^*=\mathcal{X}(\xi_T)$, where $\mathcal{X}$ satisfies $\mathcal{X}\in\mathcal{C}$ and $\E[\xi_T\X(\xi_T)]=x_0$. 
Define $\tilde{\mathcal{A}}:=\mathcal{A}$ introduced in Definition \ref{def:fuc-class} and $\delta$ is given by Eq. (\ref{eq:delta}). 
At this point, this elicitation problem is completely solved. That is, giving $X_T^*\in \Xi$ (and thus the investor is utility-based), an elicited utility function is given by $U=T_L^{-1}(X_T^*)$, and all the elicited utility functions satisfying Eq. (\ref{eq:delta}) form an affine set
$V_{\text{aff}} = \{ V \mid V = aU + b,\ a > 0,\ b \in \mathbb{R} \}$. Obviously, all these utility functions share a common absolutely risk aversion function $A:x\mapsto -\frac{U^{''}(x)}{U'(x)}$ and a common relative risk aversion function $R:x\mapsto-\frac{xU''(x)}{U'(x)}$.
ARA and RRA respectively measure how fast an investor's additional satisfaction decreases when earning a fixed absolute amount of money and a fixed percentage of wealth. The more curved (i.e., the more concave) the utility function, the larger ARA value and RRA value are, reflecting stronger risk aversion. When the utility function is linear, both ARA and RRA are zero. These functions effectively capture the investor's risk preference across different wealth levels and can be uniquely elicited. Moreover, in the traditional problem (\ref{eq:prob_intr}), given the ARA or RRA function instead of the utility function, the optimal portfolio can be derived.

While the terminal wealth is intuitive and tractable, directly specifying its explicit form for the investors remains challenging. Instead, as shown in Section \ref{sec:model}, the investors possess a notion of the satisfactory wealth value across different scenarios, which manifests as a grasp of specific probability-wealth pairs. In Section \ref{sec:fittingmethod}, we aim to fit the optimal terminal wealth $X_T^*$ of the investor using a function $X_T^k$, which is constructed by finite probability-wealth pair samples. We ensure that $X_T^k$ is the terminal wealth for a utility-based investor. 
This requirement is equivalent to seeking for a function $\mathcal{X}_T^k$ that is nonnegative, decreasing, left continuous and satisfies $\lim\limits_{\xi\to 0^+}\mathcal{X}_T^k(\xi) = \infty$ and $\E\left[\xi_T\mathcal{X}_T^k(\xi_T)\right]=x_0$. Then, we can set the fitted terminal wealth as $X_T^k=\mathcal{X}_T^k(\xi_T)$.

\begin{remark}\label{remark:nu_no}
In our framework, the fitted terminal wealth introduced in Section \ref{sec:fittingmethod} is independent of the Lagrange multipliers. In traditional problems, when we obtain the terminal wealth function, we require the Lagrange multiplier to further explore the form of the optimal terminal wealth, while in our framework, when we construct the fitted terminal wealth $X^k_T$, we directly obtain the fitted utility function $T_L^{-1}(\mathcal{X}_T^k)$.
\end{remark}

The example below is an application of the bijection, which reveals that an investor with a state-dependent utility (see \cite*{LLZ2024} or \cite*{HL2025}) can also be a utility-based investor.
\begin{example}\label{ex_B}
Let $m=1$. We consider an investor with a random benchmark, which depends on $W_T$. For example, suppose that the investor's utility function is given by $U^{W_T}(x) = \ln(x - \exp \{-\xi_T\})$ (when $x_0$ is large enough, we have $X_T^*>\exp \{-\xi_T\}$). When the market is good (i.e., when $W_T$ is large), the investor's utility reference point will be appropriately elevated in order to offer greater opportunities for pursuing a higher return. We let $r=0.05$, $T=1$, $x_0=1$, and $\theta=0.25$. Using numerical integration methods, we obtain $\nu^*=1.553$ and we let $f_l(x)=\frac{1}{x}+e^{-\frac{x}{\nu^*}}$. Assuming $V(1)=0$, according to Theorem \ref{th:TL}, the corresponding function in $\mathcal{A}$ is given by
\begin{align*}
V(x)&=\int_1^xT_2(f_l)(t)\d t=\int_1^x f_l^{-1}(t)\d t=\int_{T_2(f_l)(x)}^{T_2(f_l)(1)}f_l(t)\d t+xT_L(f_l(x))-T_L(f_l(1))\\ &=\ln\left(T_2(f_l)(1)\right)-\ln\left(T_2(f_l)(x)\right)-\nu^*\exp\left\{-\frac{T_2(f_l(1))}{\nu^*}\right\}+\nu^*\exp\left\{-\frac{T_2(f_l(x))}{\nu^*}\right\}\\ &\quad + xT_L(f_l(x))-T_L(f_l(1)).
\end{align*}
In Figure \ref{fig:dif}, we simultaneously plot the graphs of different utility functions for comparison while in Figure \label{fig:arr}, we exhibit their RRA functions.
\begin{figure}[htbp]
    \centering
    \begin{minipage}{0.48\textwidth}
        \centering
        \includegraphics[width=\textwidth]{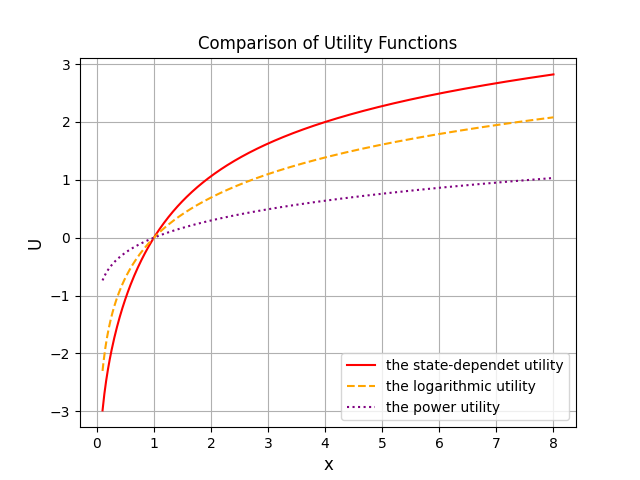}
        \caption{The specific expression of the logarithmic utility and the power utility is $U_1(x)=\ln x$ and $U_2(x)=2(x^{0.2}-1)$, respectively.}
        \label{fig:dif}
    \end{minipage}
    \hfill
    \begin{minipage}{0.48\textwidth}
        \centering
        \includegraphics[width=\textwidth]{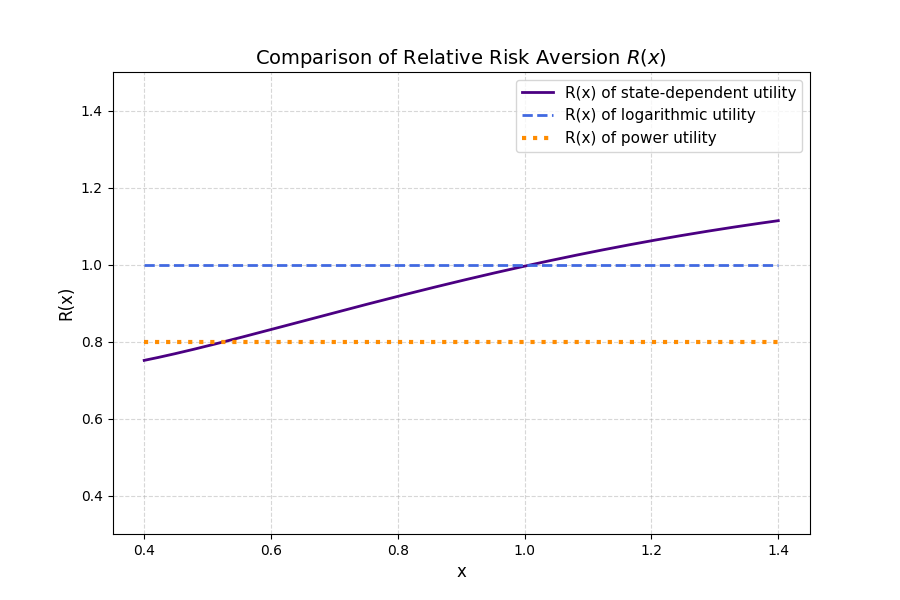}
        \caption{RRA functions accurately reflect the risk preference characteristics of different utility functions. When wealth is large, the state-dependent utility is the most risk-averse.}
        \label{fig:arr}
    \end{minipage}
\end{figure}
This example implies that we can transform the study of some state-dependent utility functions into the study of deterministic utility functions.
\end{example}

In the next section, based on the bijection, we construct a fitted terminal wealth which is considered as a function of $\xi_T$ using probability-wealth pairs. The fitted terminal wealth is the terminal wealth of a utility-based investor, and thus the corresponding utility function can be obtained using $T_L^{-1}$. 
\section{Preference-fitting Method}\label{sec:fittingmethod}
We aim to fit the optimal terminal wealth of the investor with the probability-wealth pairs and make the fitted terminal wealth a terminal wealth of a utility-based investor. Then, we can construct a fitted utility function which is the elicited utility function corresponds to the fitted terminal wealth. Then, applying the martingale-duality method, we obtain the fitted wealth process and the fitted portfolio. 
Preference-fitting presented below is effective for all investors who admit they are utility-based (i.e. profit-maximizing investors whose wealth increases as market conditions improve, according to Remark \ref{remark:utilitybased}), including those described in \cite*{KLSX1991}, \cite*{LIN2017137}, \cite*{HK2018}, \cite*{DZ2019}, \cite*{LL2024}, and so on.


We start with some clarifications. Fix $p\in(0,1)$ and define $N:=\frac{\boldsymbol{\theta}^\top \mathbf{W}_T}{||\boldsymbol{\theta}||_2\sqrt{T}}=\frac{-(r+\frac12 ||\boldsymbol{\theta}||^2_2)T-\ln (\xi_T) }{||\boldsymbol{\theta}||_2\sqrt{T}}$. Then $N$ has a standard normal distribution and acts as an indicator of the market condition; a larger $N$ means a better market condition, as stated in Remark \ref{remark:utilitybased}.
Define $\xi^p=\{\xi_T:N=\Phi^{-1}(p)\}$ as the value of $\xi_T$ at the $p$-confidence level, meaning that there is a probability of $1-p$ that the market has a rather good state. We can see from Eq. \eqref{eq_xi} that $\xi_T$ is decreasing with respect to $N$. We take the similar notation as \cite*{LLMV2024} that $d_{1}(z)=\frac{1}{-||\boldsymbol{\theta}||_2\sqrt{T}}\left(\log\left(z\right)+\left(r-\frac{||\boldsymbol{\theta}||_2^2}{2}\right)T\right)$. The investor chooses $p_0^1$ and $p_1^1$ such that $p_0^1>p_1^1$ as reference points of extremely good and extremely bad market condition, respectively. For example, she sets $p_0^1=0.999$ and $p_1^1=0.001$. We consider the market conditions are normal when $\xi_T\in \left(\xi^{p_0^1},\xi^{p_1^1}\right)$, and otherwise, we regard the market conditions as overly extreme. In this section, we are primarily concerned with normal market conditions\footnote{To simplify the subsequent analysis of the uniform convergence, considering extreme and normal market cases separately is merely technical. For the preference-fitting procedure, however, this separation is unnecessary. See Remark \ref{remark:Stepone} for details.}. The hyperbolic expression refers to the inverse of the derivative of the later Eqs. \eqref{eq:base1}-\eqref{eq:base3]} or refers to a constant function $y=h_{i+1}^k$. For example, for Eq. \eqref{eq:base3]}, the hyperbolic expression is given by $y=h_{i+1}^k+\frac{\kappa_{i+1}^k}{\xi}$ defined on some interval, where $\kappa_{i+1}^k\in \R^+$ and $h_{i+1}^k\in \R$.

In traditional research, the investor possesses a given utility function $\tilde{U}$, with which she can solve for the optimal terminal wealth $\tilde{X}_T^*=\tilde{\mathcal{X}}_T^*(\xi_T)$ using the Legendre transform, where $\tilde{\mathcal{X}}_T^*:\xi \mapsto \tilde{f}_l(\nu^*\xi)$, $\tilde{f}_l=T_L\left(\tilde{V}\right)$ and $\nu^*$ represents the Lagrange multiplier. However, the investor often finds it challenging to specify the exact form of $\tilde{U}$. 
Instead, she is much clearer about her desired wealth under various market states, which can be mathematically represented by the probability-wealth pairs $(p,y)$. The pairs $(p,y)$ indicate that ``the investor requires a return exceeding $y$ with a probability of $p$ in her terminal wealth $X_T^*$". Then, we introduce the preference-fitting procedure. Roughly speaking, assuming that $X_T^*$ is the terminal wealth for a utility-based investor, we find that the pair $(p,y)$ means $\mathcal{X}_T^*(\xi^p)=y$ because the optimal terminal wealth of an investor with a utility function can be seen as a decreasing function of $\xi_T$ according to Remark \ref{remark:X_T^*}. Following this observation, by connecting the pairs $(\xi^p,y)$ with some specific expression, we construct a fitted terminal wealth $X_T^k=\mathcal{X}_T^k(\xi_T)$ which satisfies $\mathcal{X}_T^k(\xi^p)=y$. We can derive the fitted utility function of the investor through the bijection established in Section \ref{sec:bijection}. Moreover, the fitted terminal wealth can be achieved by the fitted portfolio, which has an explicit form and is independent of the Lagrange multipliers.
 
The method is divided into three steps. In the first step, the investor directly gives a specific expression of $\mathcal{X}_T^*$ when the market conditions are extreme, namely the cases where $\xi_T\notin\left(\xi^{p_0^1},\xi^{p_1^1}\right)$ happens. In other words, the investor admits her optimal terminal wealth is given as that in \textbf{Step 1}. We will explain this requirement later and in Remark \ref{remark:Stepone}. The budget bounds in \textbf{Step 2} are necessary conditions for becoming utility-based and guaranteeing the existence of the elicited utility function. The bounds reflect the investor's understanding of the Black-Scholes model, allowing them to strike a balance between pursuing wealth and mitigating risk. In the rest of this section, we first present the specific steps and then provide further explanations.\\
\textbf{Step 1. Specify the optimal terminal wealth under extreme market conditions.}\\
Choose $p_0^1$ and $p_1^1$ such that $p_0^1>p_1^1$, and define $y_0^1=\frac{\kappa_0^1}{\xi^{p_0^1}}$ and $y_1^1=\frac{\kappa_1^1}{\xi^{p_1^1}}$, where $\kappa_0^1>0$ and $\kappa_1^1>0$ are chosen to fulfill $y_0^1\geq y_1^1$ and the following conditions
\begin{equation*}\small
\begin{aligned}
&\kappa_0^1(1-p_0^1)+\kappa_1^1p_1^1 +y_1^1\E\left[\xi_T 
\id_{\{\xi_T\in (\xi^{p_0^1},\xi^{p_1^1})\}}\right] \leq x_0 < \kappa_0^1(1-p_0^1)+\kappa_1^1p_1^1+y_0^1\E\left[\xi_T\id_{\{\xi_T\in (\xi^{p_0^1},\xi^{p_1^1})\}}\right],
\end{aligned}
\end{equation*}
and
assume the optimal terminal wealth $\mathcal{X}_T^*(\xi)=\frac{\kappa^1_0}{\xi}$ when $\xi \in (0,\xi^{p_0^1})$ and $\mathcal{X}_T^*(\xi)=\frac{\kappa^1_1}{\xi}$ when $\xi \in (\xi^{p_1^1},\infty)$. For example, the investor can set $p_0^1=0.999$, $p_1^1=0.001$ and $\kappa_0^1=\kappa_0^2=x_0$, where $x_0$ is the value of initial asset.\\
\textbf{Step 2. Give the probability-wealth pairs under normal market conditions.}\\
Suppose that $p_0^n,p_1^n,\cdots , p_n^n$  and  $y_0^n,y_1^n,\cdots , y_n^n$ ($n\in \N^*$) are well defined. Choose $p_{n+1}$ satisfying $p_{n+1} \notin \{p_0^n,p_1^n,\cdots ,p_n^n\}$ and $p_{n+1}\in (p_1^1,p_0^1)$. Arrange $p_{n+1}$, $ p_0^n$, $p_1^n,\cdots,p_n^n$ in descending order and rename them sequentially as $p_0^{n+1},\cdots,p_{n+1}^{n+1}$. Obviously, we have $p_0^{n+1}=p_0^n=p_0^1$ and $p_{n+1}^{n+1}=p_n^n=p_1^1$. Then, there exists $j_{n+1}\in \N^*$, $0<j_{n+1}<n+1$ such that $p_{n+1}=p_{j_{n+1}}^{n+1}$. Define $y_i^{n+1}=y_i^n$ for $0\leq i<j_{n+1}$ and $y_i^{n+1}=y_{i-1}^n$ for $j_{n+1}<i\leq n+1$. Then choose  $y_{j_{n+1}}^{n+1}\in[ y_{j_{n+1}+1}^{n+1},y_{j_{n+1}-1}^{n+1}]$ satisfying
\begin{equation}\label{eq:selections1}
\begin{aligned}
\sum_{i=0}^{n-1} y_{i}^{n+1}\E\left[\xi_T\id_{\{\xi_T\in (\xi^{p_i^{n+1}},\xi^{p_{i+1}^{n+1}})\}}\right]+\kappa_0^1(1-p_0^1)+\kappa_1^1p_1^1> x_0,
\end{aligned}
\end{equation}
and
\begin{equation}\label{eq:selections2}
\begin{aligned}
\sum_{i=0}^{n-1} y_{i+1}^{n+1}\E\left[\xi_T\id_{\{\xi_T\in (\xi^{p_i^{n+1}},\xi^{p_{i+1}^{n+1}})\}}\right]+\kappa_0^1(1-p_0^1)+\kappa_1^1p_1^1\leq x_0.
\end{aligned}
\end{equation}
We refer to the right-hand side of Eqs. \eqref{eq:selections1} and \eqref{eq:selections2} as the $(n+1)$-th budget lower bound and $(n+1)$-th budget upper bound, respectively.\\
\textbf{Step 3. Connect the pairs with the hyperbolic expressions.}\\
Choose $k\in \N^* $. Repeating the operations in Step 2, we obtain $(p_i^k, y_i^k)$, $i=0,1,\cdots,k$. Then, using a hyperbolic expression to connect $(\xi^{p_i^{k}},y_i^k)$ and $(\xi^{p_{i+1}^{k}},y_{i+1}^k)$ for every $i=0,1,\cdots,k-1$, we obtain a function $\mathcal{X}^k_T$ defined on $\left[\xi^{p_0^k},\xi^{p_k^k}\right]$. Then, set $\mathcal{X}^k_T(\xi)=\mathcal{X}_T^*(\xi)$ when $\xi \in(0,\xi^{p_0^k})\cup (\xi^{p_k^k},\infty)$ according to Step 1. This procedure needs to satisfy two requirements: first, the connected function $\mathcal{X}_T^k$ is in the function class $\mathcal{C}$; and second, the budget constraint $\E\left[\xi_T\mathcal{X}_T^k(\xi_T)\right]=x_0$ holds, where $x_0$ is the value of initial asset. Therefore, we have completed defining $\mathcal{X}_T^k$ for any $\xi\geq 0$. Finally, we define $X_T^k=\mathcal{X}_T^k(\xi_T)$.

The schematic diagram of preference-fitting method is shown in Figure \ref{fig:pair}.
\begin{lemma}\label{lemma_existences}
One can guarantee the existence of $\kappa_0^1$ and $\kappa_1^1$ in Step 1, $y_{j_{n+1}}^{n+1}$ in Step 2, and the hyperbolic expressions in Step 3.
\end{lemma}
\begin{proof}
Assume that $\xi^{p_0^{n}},\cdots,\xi^{p_n^n}$ and $y_0^n,\cdots,y_n^n$ are given. By straightforward calculation, we find that the budget upper bound is not lower than $y_{j_{n+1}+1}^{n+1}$ and is not less than the budget lower bound, which is not greater than $y_{j_{n+1}-1}^{n+1}$. Therefore, the existence of $y_{j_n+1}^{n+1}$ holds by induction. The existences of $\kappa_0^1$ and $\kappa_1^1$ and the hyperbolic expressions are evident due to the intermediate value theorem.

\end{proof}
For \textbf{Step 1}, the key is to determine the expression for $X_T^*$ under extreme market conditions. In this step, when $(\xi^{p_0^1},y_0^1)$ is defined, we connect it with the point $(0,\infty)$ using a hyperbolic expression $y=\frac{\kappa_0^1}{\xi}$ to define $\mathcal{X}_T^*$ and we treat $(\xi^{p_1^1},y_1^1)$ and $(\infty,0)$ in the similar way. 
This step is intended to facilitate the subsequent discussion of the uniform convergence and it is not necessary.
On one hand, directly assuming the form of the terminal wealth under extreme market conditions has little impact due to the extremely low probability of such scenarios. In the studies we reference, these cases also contribute negligibly to the value of $\E[\xi_TX_T^*]$. On the other hand, \cite*{DB1954} demonstrates that the utility of large wealth in people's minds generally takes a logarithmic form through empirical analysis. This form corresponds to the hyperbolic expression introduced in \textbf{Step 1}.
We will add some supplementary descriptions in Remark \ref{remark:Stepone}. \textbf{Step 1} also plays a role in the initialization of the budget upper and lower bounds. In the subsequent steps, $\kappa_0^1$ and $\kappa_1^1$ only affect the values of the budget upper and lower bounds but not the feasibility of the steps. The inequalities in \textbf{Step 1} are special forms of Eqs. \eqref{eq:selections1} and  \eqref{eq:selections2}.


For \textbf{Step 2}, the discussion preceding Remark \ref{remark:nu_no} suggests that the terminal wealth of a utility-based investor can be considered as a function belong to the class $\mathcal{C}$ satisfying the budget constraint. To guarantee that the fitted terminal wealth satisfies the two conditions, the investor strategically positions her anticipated return within the range of the budget upper and lower bounds. Through direct calculation we have Eq. (\ref{eq:selections1}) is equivalent to 
\begin{align*}
y_{j_{n+1}}^{n+1}>&\frac{x_0-\kappa_0^1(1-p_0^1)-\kappa_1^1p_1^1-e^{-rT}\sum_{\substack{i \ne j_{n+1}\\0\leq i \leq n }}y_i^{n+1}\left[\Phi\left(d_{1}\left(\xi^{p_i^{n+1}}\right)\right)-\Phi\left(d_{1}\left(\xi^{p_{i+1}^{n+1}}\right)\right)\right]}{e^{-rT}\left[\Phi\left(d_{1}\left(\xi^{p_{j_{n+1}}^{n+1}}\right)\right)-\Phi\left(d_{1}\left(\xi^{p_{j_{n+1}+1}^{n+1}}\right)\right)\right]},
\end{align*}
while Eq. (\ref{eq:selections2}) is equivalent to
\begin{align*}
y_{j_{n+1}}^{n+1} &\leq  \frac{x_0-\kappa_0^1(1-p_0^1)-\kappa_1^1p_1^1-e^{-rT}\sum_{\substack{i \ne j_{n+1}\\0\leq i \leq n }}y_{i+1}^{n+1}\left[\Phi\left(d_{1}\left(\xi^{p_{i}^{n+1}}\right)\right)-\Phi\left(d_{1}\left(\xi^{p_{i+1}^{n+1}}\right)\right)\right]}{{e^{-rT}\left[\Phi\left(d_{1}\left(\xi^{p_{j_{n+1}}^{n+1}}\right)\right)-\Phi\left(d_{1}\left(\xi^{p_{{j_{n+1}}+1}^{n+1}}\right)\right)\right]}}.
\end{align*}
As discussed in Section \ref{sec:bijection}, the requirement $\E[\xi_TX_T^n]=x_0$ reflects a trade-off: while investors pursue higher returns, they also recognize that higher returns come with greater risk.  
This is reflected in the fact that an increase in $y_i^n$ leads to a decrease in the $n+1$-th budget upper bound.
We emphasize that the upper and lower bounds contain all possible values of wealth at a specified probability level and are endogenous to the Black-Scholes model, which are independent of the martingale-duality method. In other words, the bound constraints are necessary conditions for an investor to become utility-based. Thus, we consider the preferences of a utility-based investor to be inadmissible if they do not satisfy these requirements of the bound. For instance, an investor aiming to achieve a return of 100 with a probability of 0.99, starting with an initial asset of 1, would be deemed to have unrealistic anticipation in the case of the same parameter settings as in Section \ref{sec:VaR} later. Eqs. (\ref{eq:selections1}) and (\ref{eq:selections2}) mathematically characterizes the investor's perception of the Black-Scholes market, namely, that higher returns come at the cost of higher risk. This trade-off is prevalent across different market settings.

For \textbf{Step 3}, we find a continuous function in class $\mathcal{C}$ passing through the points determined by the pairs and satisfying the budget constraint. This terminal wealth retains the preference characteristics of utility-based investors because $\mathcal{X}_T^k \in \mathcal{C}$. 
We can use the martingale-duality method to derive the fitted portfolio which achieves the fitted terminal wealth through the SDE (\ref{eq_notct}). 
Moreover, \citet*{LLMV2024} reveal that the equation $\E[\xi_TX_T^k]=\E[\xi_T\X_T^k(\xi_T)]=x_0$ is equivalent to an explicit equation with multiple parameters.
Hence, the hyperbolic expressions are not unique\footnote{the existence is guaranteed by \textbf{Step 2}.}. Investors first choose the pairs as their core focus, then calibrate the parameters in the hyperbolic expressions to forms that meets their satisfaction and satisfies the budget constraints $\E[\xi_T\X_T^k(\xi_T)]=x_0$.

We emphasize that, because the budget constraint is endogenous to the model, the preference-fitting method can construct any attainable optimal terminal wealth process (under normal market conditions, if we specify the extreme market conditions).

The preference-fitting method eliminates the necessity of discussing the Lagrange multiplier. In fact, it is evident from the steps that the fitted terminal wealth $X_T^k$ is independent of $\nu^*$. 
Then, using the martingale-duality method, we derive the fitted wealth process as $X_t^k=\xi_t^{-1}\E\left[\xi_TX_T^k|\F_t\right]:=\mathcal{X}^k(t,\xi_t)$ and the fitted portfolio as $\boldsymbol{\pi}_t^k=-\xi_t \frac{\partial \mathcal{X}^k(t,\xi_t)}{\partial \xi_t}(\boldsymbol{\sigma}^\top)^{-1}\boldsymbol{\theta}$, which are independent of $\nu^*$; see Section \ref{sec:phara} later. 
\begin{figure}
    \centering
    \includegraphics[width=0.65\linewidth]{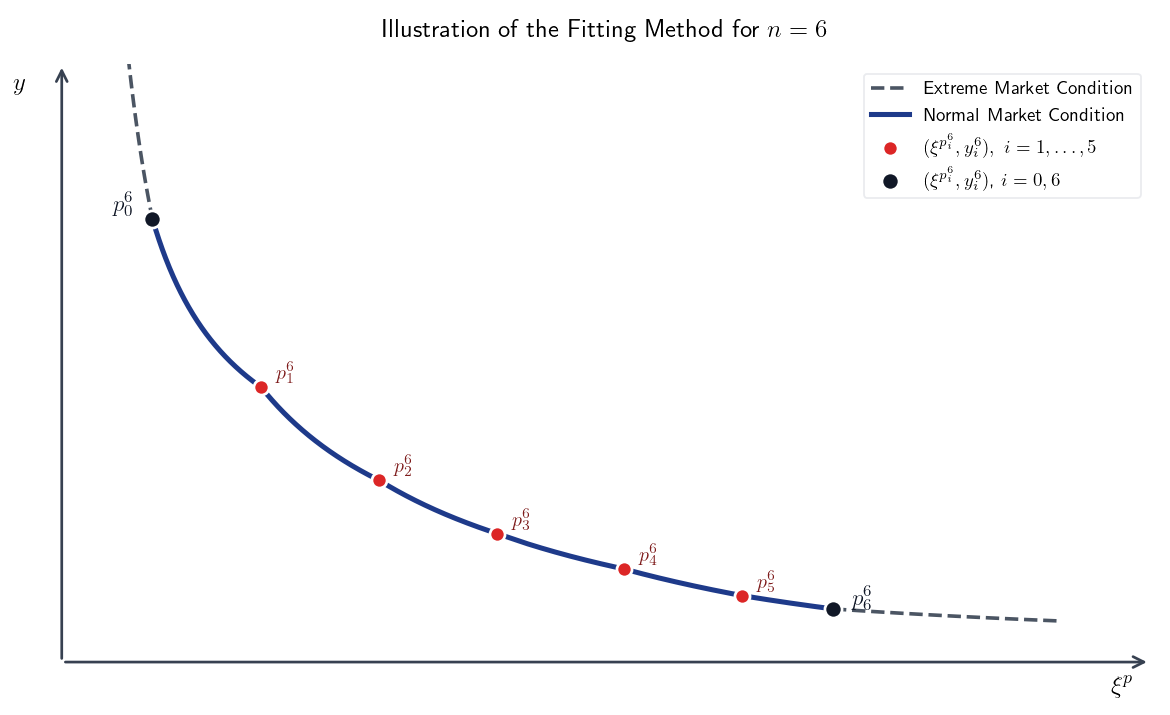}
    \caption{This figure depicts the approximate shape of the fitted terminal wealth when $n=6$. The investor targets returns at six probability levels; once the corresponding pairs are determined in turn, she adopts a suitable hyperbolic specification to connect these points.}
    \label{fig:pair}
\end{figure}
\begin{definition}\label{def:fitXpi}
For a fitted terminal wealth $X_T^k$, we define the fitted utility function as $U_k=T_L^{-1}(X_T^k)$, the fitted wealth process as $X_t^k=\xi_t^{-1}\E\left[\xi_TX_T^k|\F_t\right]$, and the fitted portfolio as $\boldsymbol{\pi}_t^k=-\xi_t \frac{\partial \mathcal{X}^k(t,\xi_t)}{\partial \xi_t}(\boldsymbol{\sigma}^\top)^{-1}\boldsymbol{\theta}$, where $k\in \N^*$ and $T_L^{-1}$ is the mapping given by Definition \ref{def:bij}.
\end{definition}
In practice, investors may select several focal scenarios (corresponding to different probability levels $p$) and specify their expected wealth based on investment objects and psychological expectations, thereby constructing the corresponding fitted portfolios. We believe this fitted portfolio is already practically usable. Naturally, in order to obtain more accurate results of preference and to demonstrate the effectiveness of the preference-fitting method, we require that the fitted terminal wealth converges as the fitting accuracy increases. Moreover, we require the fitted wealth process, the fitted portfolio and the fitted utility function converge to the optimal wealth process, the optimal portfolio and the elicited utility function, respectively. Indeed, if $\max\limits_{\{i=0,1,\cdots, n-1\}} \{|\xi_i^n-\xi_{i+1}^n| \} \to 0$, as $n\to \infty$, we can verify various kinds of convergences of the fitted wealth process and the fitted portfolio to the optimal wealth process and the optimal portfolio, respectively, using the PHARA approximation method; see Section \ref{sec:phara}.
\begin{remark}\label{remark:Stepone}
We provide further explanation for \textbf{Step 1} which is an optional operation rather than a mandatory step. We can revise \textbf{Step 1} as follows. The investor provides $\{\xi^{p^k_n}\}_{0\leq n \leq k}$ lying in the interval $(0, \infty)$, and she only needs to ensure that the budget lower bound constraint is satisfied when choosing the corresponding $y$ values (as the upper bound is always $\infty$). Additionally, in order to guarantee the convergences, we require that $\xi^{p_0^k}\to 0$, $\xi^{p_k^k}\to \infty$, and $\max\limits_{1\leq n\leq k}\{\xi^{p_n^k}-\xi^{p_{n-1}^k}\}\to 0$ as $k \to \infty$. However, if \textbf{Step 1} is modified in this manner, the convergence of the method weakens (Section \ref{sec:phara}), as no upper bound is imposed on the constraints. Theoretically, the investor's behavior can be highly pathological: for instance, they may forgo income under normal market conditions in exchange for a extremely small probability of securing an extraordinary payoff.
\end{remark}

The preference-fitting method presents several advantages featuring intuitiveness and data utilization efficiency. First, drawing on insights from behavioral economics surveys, this approach aligns more closely with the cognitive process of the investor, relying merely on finite samples. Second, we fit the terminal wealth and obtain the fitted portfolio with merely finite samples rather than continuous functional expression. Third, the utility function derived from this method circumvents the need for Lagrange multipliers. Instead, we propose budget bounds to guide and assist investors in allocating wealth, thereby enhancing their operational flexibility. 
Forth, investors can reduce asset fluctuations by evenly distributing the $y$-values corresponding to different p-values. According to Section 7 in \cite*{LLMV2024}, it is the high frequency of discontinuous points (arising from identical $y$-values) that leads to high fluctuations.

\begin{remark}\label{remark:beha}
Research in behavioral economics indicates that investors often adopt a psychological benchmark as a reference point and display loss aversion (\cite*{kahneman1979d}). In our approach, the investors may have tendency to establish a higher target wealth value under the poor market conditions to mitigate the risk of bankruptcy. The budget upper bound indicates the highest wealth value they can choose, and if this does not align with their preferences, they tend to lower their $y$ values during more favorable market states to enhance the upper bound, thus achieving a more robust portfolio.
\end{remark}

In Subsection \ref{subsec:app}, we propose the PHARA-approximation approach. In Subsection \ref{subsec:fittingrate}, we consider the following scenario: the investor has a preferred optimal terminal wealth $X_T^*$ (though its exact distribution is unknown), and sequentially provides pairs $(p, y)$ satisfying $\mathbb{P}(X_T^* > y) = p$. From these, we obtain the fitted terminal wealth, fitted wealth process, fitted utility function, and fitted portfolio. We then show that, as the fitting accuracy increases, these fitted elements converge in various senses. We demonstrate the explicit solutions of fitted wealth process and fitted portfolio in Subsection \ref{subsec:expl}.
\section{PHARA Utility Approximation}\label{sec:phara}
In this section, we present the PHARA approximation method to derive explicit expressions for the fitted wealth process and the fitted portfolio and then demonstrate their convergences. For an investor with a utility function $V$, the PHARA-approximated wealths and the PHARA-approximated portfolios defined in Definition \ref{def:fitting} 
(which differ from the definitions of the fitted wealth and the fitted portfolio in Section \ref{sec:fittingmethod} to allow a broader range of applications for the PHARA approximation; see Section \ref{sec:ex_hyper}.) converge to the optimal wealth and optimal portfolio, respectively. We first present the PHARA-approximation because the analyses of its form and convergence are more general.
\subsection{The Approximation Procedure}\label{subsec:app}
To begin, we introduce the HARA base.
\begin{definition}
We define two kinds of HARA bases on interval $[s_0,t_0)$ as follows:\\
(1) The power base
\begin{equation}\label{eq:base1}
B(x)=[a(x-u)^\gamma+b]\id_{\{s_0\leq x< t_0\}},
\end{equation}
where $a>0,u\leq s_0,\gamma\in (0,1],b \in \R$, or $a<0,u\leq s_0,\gamma\in(-\infty,0),b\in \R$, and
\begin{equation}\label{eq:base2}
B(x)=[a(u-x)^\gamma+b]\id_{\{s_0\leq x< t_0\}},
\end{equation}
where $a<0,u\geq t_0,\gamma\in (1,\infty),b \in \R$.\\
(2) The logarithmic base
\begin{equation}\label{eq:base3]}
B(x)=[a\ln (x-u)+b]\id_{\{s_0\leq x< t_0\}},
\end{equation}
where $a>0,u\leq s_0,b\in \R$. For the sake of simplicity in calculations, we denote the logarithmic utility corresponding to $\gamma = 0$.
\end{definition}

An explicit solution for the optimal portfolio can be obtained only when the investor's utility function conforms to specific forms. Prominent examples of such utility functions include PHARA, which contains power utility, logarithmic utility, and exponential utility, as outlined in \cite*{LLMV2024}.

Given $x_1<x_2$, $k_{x_1}\geq k_{x_2}$ and $v_1 \in \R$, we can use a power base $B(x)$ (which may not be unique) on the interval $[x_1,x_2)$ such that $B'_+(x_1)=k_{x_1},B'_-(x_2)=k_{x_2}$, and $B(x_1)= v_1$. This assertion can be verified through straightforward calculation. If $k_{x_1}> k_{x_2}$, we can also find a logarithm base that satisfies these conditions. This observation motivates us to introduce the following definition.
\begin{definition}\label{def:V_n}
Define the approximated utility function $V_n\in \mathcal{A}$ for $V \in \mathcal{A}$ and $n \in \mathbb{N}^*$ as follows. For $n\in \N^*$, let $T(n)>0$ and $M(n)\in \N^*$. Divide the interval $[0,T(n)]$ into $M(n)$ segments. We require that $T(n) \to \infty$ as $n\to\infty$ and that the maximum length of the $M(n)$ parts, denoted as $d_n$, converges to 0 as $n \to \infty$. For the $i$-th segment $[s_i^n, t_i^n)$, where $1\leq i \leq M(n)$, we construct a HARA base $B_i^n(x) = \left[a_i^n (x - u_i^n)^{\gamma_i^n} + b_i^n \right]\id_{\{ x\in[s_i^n,t_i^n)\}}$ or alternatively, $B_i^n(x)=[a^n_i \ln (x-u_i^n)+b_i^n]\id_{\{ x\in[s_i^n,t_i^n)\}}$ such that $(B_i^n)'_+(s_i^n)=V'_+(s_i^n)$ and $(B_i^n)'_-(t_i^n)=V'_-(t_i^n)$ hold for any $1\leq i \leq M(n)$ while $B_i^n(s_i^n) = B_{i-1}^n(s_i^n)$ holds for any $2\leq i \leq M(n)$. In addition, we define \begin{small}$B_{M(n)+1}^{n}(x)= \left( \ln (x-u^n_{M(n)+1}) + b_{M(n)+1}^n \right)\id_{\left\{ x\in \left[t_{M(n)}^n,\infty \right) \right\}}$\end{small}. We require \begin{small}$B_{M(n)+1}^{n}\left(t_{M(n)}^n\right)=V\left(t_{M(n)}^n\right)$\end{small} and \begin{small} $\left(B_{M(n)+1}^{n}\right)'_{+}\left(t^n_{M(n)}\right)=V'_+\left(t_{M(n)}^n\right)$\end{small}. Define $V_n(x) := \sum\limits_{i=1}\limits^{M(n)+1}B_i^n(x)$, i.e., 
define
\begin{equation}\label{fuc_phara}
    V_n(x) = \left\{
    \begin{aligned}
    &\left[a_i^n (x - u_i^n)^{\gamma_i^n} + b_i^n \right]\id_{\{ x\in[s_i^n,t_i^n)\}}, \quad \gamma_i \in (0,1],\ x \in [s_i^n, t_i^n),\ i=1,\cdots,M(n); \\
    &  a_i^n\ln\left(x - u_i^n\right)+ b_i^n \id_{\{ x\in[s_i^n,t_i^n)\}}, \quad \gamma_i = 0, \ x \in [s_i^n, t_i^n),\ i=1,\cdots,M(n);\\
    &\left[\ln (x-u^n_{M(n)+1}) + b_{M(n)+1}^n \right]\id_{\left\{ x\in \left[t_{M(n)}^n,\infty \right) \right\}}, \quad   x\in[t_{M(n)+1},\infty).
    \end{aligned}
    \right.
\end{equation}
We refer to any $V_n$ constructed in this manner as a PHARA-approximated utility function of $V$.
\end{definition}
In Section \ref{sec:fittingmethod}, we first obtain the fitted terminal wealth $X_T^k=\X(\xi_T)$, and the fitted utility function $V_k=T_L^{-1}(\X)$ takes the form of an approximated utility.
\begin{remark}
Without loss of generality, we may assume $V_n(1) = 0$ to ensure $V_n \in \mathcal{A}$. Otherwise, this can be achieved by vertically
shifting the graph of $V_n$, which has no impact on the subsequent analysis.
\end{remark}
\begin{remark}\label{remark:domain}
The domain of the various functions in Definitions \ref{class} and \ref{def:V_n} is $[0,\infty)$ and we can easily extend it to $[a, \infty)$ or $(- \infty, \infty)$ for $a\in \R$. For example, the definition of $\mathcal{A}$ can be presented as the set of all concave and strictly increasing functions $V$ with domain $(-\infty, \infty)$, $V(1) = 0$ and $\lim_{x\to\infty}V'_+(x)=0$. When $V$ is defined on $(- \infty, \infty)$, the divisions in the construction of $V_n$ should be made within the interval $[-T_1(n), T_2(n)]$, and we use a logarithmic base defined on $(-\infty, -T_1(n))$ with the form $B_1^n(x)=[-\ln(u_1^n-x)-b_1^n]\id_{\{ x\in(-\infty,-T_1(n))\}}$, where $u_1^n>-T_1(n),b_1\in \R$ such that $B_1(-T_1(n))=V(-T_1(n))$ and $(B_1)'_-(-T_1(n))=V'_-(-T_1(n))$ hold to formulate $V_n$. 
\end{remark}

For $V$ and the corresponding $V_n \in \mathcal{A}$, where $n\in \N^*$, we define $f=T_L(V)$ and $f_n=T_{L}(V_n)$. In addition, we define $g: \nu \mapsto \E [\xi_Tf(\nu \xi_T)]$ and $g_n:  \nu \mapsto \E [\xi_Tf_n(\nu \xi_T)]$ as functions of $\nu$. For Problem (\ref{eq:prob_intr}), because $V_n$ is a PHARA utility for any $n\in \N^*$, the optimal portfolio of an investor with utility function $V_n$ has an explicit form.

The second condition in Assumption \ref{ass:L2} is equivalent to $\E \[(\xi_Tf(\nu \xi_T))^2\] < \infty$ for any $\nu>0$ in this section. The following results rely on this key assumption.

For an investor with utility function $V_n$, the existence of the optimal portfolio is guaranteed by the following lemma \ref{nu_covergence}.
\begin{lemma}\label{nu_covergence}
For $V \in \mathcal{A}$ and $V_n$ in Definition \ref{def:V_n}, we have the following results.\\
(1) The series of functions $ \{V_n\}_{n \geq 1} $ converges uniformly to $V$ on any closed subset of $(0, \infty)$.\\ 
(2) $\lim\limits_{n\to \infty}\E\left[(\xi_Tf_n(\nu\xi_T)-\xi_Tf(\nu\xi_T))^2\right]=0$. Thus $\E [(\xi_Tf_n(\nu \xi_T))^2]$ has a uniform upper bound $M_\nu$ for any $\nu>0$.\\
(3) Each of the equations $g(\nu) =x_0$ and $g_n(\nu)=x_0$, $n\in\N^*$, admits a unique solution. Denote the unique solutions by $\nu^*$ and $\nu_n^*$, respectively. Then we have $\nu_n^* \to \nu^*$ as $n \to \infty$.
\end{lemma}
\begin{proof}
See Appendix \ref{l4.1}.
\end{proof}
For an investor with a utility function $V$, we denote by $X_T^*$ the optimal terminal wealth, by $\{X_t^*\}_{0\leq t\leq T}$ the optimal wealth and by $\{\boldsymbol{\pi}_{t}^*\}_{0\leq t\leq T}$ the optimal portfolio. According to Theorem 6.3 in \cite*{KLSX1991}, for investor with the utility function $V_n$, the optimal terminal wealth, denoted by $X_{n,T}^*$, the optimal wealth process, denoted by $\{X_{n,t}^*\}_{0 \leq t \leq T}$ and the optimal portfolio, denoted by $\{\boldsymbol{\pi}_{n,t}^*\}_{0\leq t \leq T}$, exist. Moreover, we have $X_{n,T}^*=f_n(\nu^*_n\xi_T)$ and $X_{n,t}^*=\xi_t^{-1}\E[f_n(\nu^*_n\xi_T)|\F_t]$, where $\nu^*_n$ is the Lagrange multiplier satisfying $\E[\xi_Tf_n(\nu^*_n\xi_T)]=x_0$. Based on this, we give the following definition.
\begin{definition}\label{def:fitting}
For an investor with a utility function $V$ and $n\in \N^*$, suppose that the PHARA-approximated utility function is $\{V_n\}_{n\in \N^*}$. We denote by $X_{n,T}^*$ the PHARA-approximated terminal wealths, by $\{X_{n,t}^*\}_{0 \leq t \leq T}$ the PHARA-approximated optimal wealths and by $\{\boldsymbol{\pi}_{n,t}^*\}_{0\leq t \leq T}$ the PHARA-approximated optimal portfolios.
\end{definition}
We now shift our attention to the convergence of the PHARA-approximated optimal wealths and the PHARA-approximated portfolios. The main result is presented in Theorem \ref{th:co-X} and the following lemmas are crucial in the proof. The proofs of the lemmas are relegated to Appendix \ref{appendix_A}.
\begin{lemma}\label{lemma:L^2-co}
For $\nu^*$ and $\nu^*_n$ defined in Lemma \ref{nu_covergence}, we have $\E\left[(\xi_Tf_n(\nu_n^*\xi_T)-\xi_Tf(\nu^*\xi_T))^2\right]\to 0$ as $n\to\infty$.
\end{lemma}
\begin{lemma}\label{prop:1} 
For any $r_1\in \R,r_2 \in [0,\infty),r_3\in[1,2)$ and $t\in[0,T)$ and for fixed $t>0$ and $\xi>0$, we have
$$\left|\int_{0}^{\infty}x^{r_1}|\ln x|^{r_2}\left(f(\nu^*\xi x)\right)^{r_3}\d F_{t}(x)\right|<\infty,$$ where $F_t$ is the conditional  distribution function of $\frac{\xi_T}{\xi_t}$ given $\F_t$. Moreover, $$\left\{\left|\int_{0}^{\infty}x^{r_1}|\ln x|^{r_2}\left(f_n(\nu^*_n\xi x)\right)^{r_3}\d F_{t}(x)\right|\right\}_{n\in \N^*}$$ has a uniform upper bound. 
\end{lemma}
We will use Lemma \ref{prop:1} to ensure the property of uniform integrability in the proof of Theorem \ref{th:co-X}.
\begin{lemma}\label{prop:2}
The function defined on 
$(0, \infty) \times (0, T)$ by 
\begin{equation}\label{eq:K(x)}
K(x,t):=\frac{1}{\sqrt{2\pi}||\boldsymbol{\theta}||_2\sqrt{T-t}}\exp\left\{-\frac12\left(\frac{\ln x+\left(r+\frac{||\boldsymbol{\theta}||_2^2}{2}\right)(T-t)}{||\boldsymbol{\theta}||_2\sqrt{T-t}}\right)^2\right\}
\end{equation}
is $C^{\infty}$ (i.e., derivative functions of any order are continuous) with respect to both $t$ and $x$. Moreover, we have $$
\frac{\d^n}{\d x^n}K(x,t)=K(x,t)\sum_{i\leq\tau_n}c_i^nx^{\alpha_i^n}(\ln x)^{\beta_i^n},$$
where $\tau_n\in \N^*$ is determined by $n$, and $\alpha_i^n \in \mathbb{Z},\beta_i^n \in \N$ and $c_i^n$ are constants that do not depend on $x$.
\end{lemma}
Lemma \ref{prop:2}  can be directly proved by induction, and thus we skip the proof. Lemma \ref{lemma:twice-dif} below reveals that the optimal portfolio is in fact a function of the time and the pricing kernel, thus transforming the stochastic convergence problem into a convergence problem for real functions.
\begin{lemma}\label{lemma:twice-dif}
Under Assumption \ref{ass:L2}, we have the following two results:\\
(1) There exists a two-variable twice continuously differentiable function $\mathcal{X}^*$ defined on $[0,T) \times [0,\infty)$ such that the optimal wealth process $X_t^*=\xi_t^{-1}\E[\xi_Tf(\nu^*\xi_T)|\F_t]=\mathcal{X}^*(t,\xi_t)$.\\
(2) The optimal portfolio is given by $\boldsymbol{\pi}_t^*=-\xi_t \frac{\partial \mathcal{X}^*(t,\xi_t)}{\partial \xi_t}(\boldsymbol{\sigma}^\top)^{-1}\boldsymbol{\theta} $.
\end{lemma}
\begin{remark}
The correctness of Lemma \ref{lemma:twice-dif} is attributed to the smoothness of function $K$ defined by Eq. (\ref{eq:K(x)}). According to the Lebesgue decomposition, the function $f$ can be decomposed into a continuous part and a jump part, and the function $K$ somewhat ``smoothens out" the impact of the jump part of $f$. Additionally, Assumption \ref{ass:L2} is also vital, as it not only ensures the applicability of the martingale representation theorem but also provides a foundation for the use of the dominated convergence theorem (abbr. DCT) in the proof.
\end{remark}
For a matrix $\boldsymbol{\sigma}$, we define $\left|\left|\cdot \right|\right|_2$ as $||\boldsymbol{\sigma}||_2=\sqrt{tr(\boldsymbol{\sigma}\boldsymbol{\sigma}^\top)}$. For a random vector $\boldsymbol{\pi}_t = (\pi_t^1,\pi_t^2,\cdots,\pi_t^m)$, the notation $\boldsymbol{\pi}_{n,t} \xrightarrow{L^r(\Omega;\R^m)} \boldsymbol{\pi}_t$ ($r\in [1,2)$) means that each component of the random vector $\boldsymbol{\pi}_{n,t}$ converges to the corresponding component of $\boldsymbol{\pi}_{t}$ in $L^r(\Omega,\R)$ (which is simply denoted by $L^r(\Omega)$). For a fixed $t\in [0,T]$, we define $\boldsymbol{\pi}_n \xrightarrow{L^1(\Omega \times [0,T])} \boldsymbol{\pi}$ and $\boldsymbol{\pi}_{n,t} \xrightarrow{a.s.} \boldsymbol{\pi}_t$ in the similar way. To prove the a.s. convergence of a random variable, such as the fitted terminal wealth, we typically treat it as a real-valued function of $\xi_t$ and demonstrate that this real function converges a.e..
\begin{theorem}\label{th:co-X}
We have the convergence results as follows.\\
(1) For every $r\in[1,2)$, the optimal wealth process satisfies  $X_{n,t}^* \overset{L^r(\Omega)}{\underset{\text{a.s.}}{\longrightarrow}} X_t^*$,\   a.e.\  $t\in[0,T]$ and $X_{n}^* \xrightarrow{L^r(\Omega \times [0,T])}  X^*$. \\
(2) For every $r\in[1,2)$, the optimal portfolio satisfies  $\boldsymbol{\pi}_{n,t}^* \overset{L^r(\Omega)}{\underset{\text{a.s.}}{\longrightarrow}} \boldsymbol{\pi}_t^*$\ a.e. \ $t\in[0,T)$ and $\boldsymbol{\pi}_{n}^{*} \xrightarrow{L^1(\Omega \times [0,T])}  \boldsymbol{\pi}^{*}$.
\end{theorem}
\begin{proof}
Without loss of generality, we only prove the result  when $r=1$. We prove the following various senses of convergence one by one. We first prove that
 $X_{n,T}^* \overset{\text{a.s.}}{\longrightarrow} X_T^*$ holds.

Noting that $\nu^*$ is a deterministic constant, we suppose that  $x= \nu^*\xi_T(\omega)$ is a continuous point of $f \in \mathcal{C}$, and all of these $\omega$ consist of a set of probability measures $1$ because  the set of discontinuities of $f$ forms at most a countable set. For any $0<\epsilon<f(\nu^*\xi_T)$, there exists $\delta_1>0$ such that for any $k_0\in (\nu^*\xi_T-\delta_1,\nu^*\xi_T+\delta_1)$ and $k_0>0$, we have $$f(k_0) \in (f(\nu^*\xi_T)-\epsilon,f(\nu^*\xi_T)+\epsilon).$$ There exists $0<\delta_2<\delta_1$ such that for any $k_1\in (\nu^*\xi_T-\delta_2,\nu^*\xi_T+\delta_2)$, we have 
$$
f(k_1) \in \(f(\nu^*\xi_T)-\frac{\epsilon}{2},f(\nu^*\xi_T)+\frac{\epsilon}{2}\).
$$ 
By Lemma \ref{nu_covergence}, there exists $N_1\in \N^*$ such that for $n>N_1$, we have $\nu^*_n\xi_T(\omega) \in (\nu^*\xi_T-\delta_2,\nu^*\xi_T+\delta_2)$. Let $N_2$ be the integer satisfying the partition radius $d_{N_2} < \epsilon/4$. Then for any $n > \max\{N_1,N_2\}$, there exist $s_i^n\in (f(\nu^*\xi_T)-\epsilon,f(\nu^*\xi_T)-\epsilon/2)$ and $s_j^n \in (f(\nu^*\xi_T)+\epsilon/2,f(\nu^*\xi_T)+\epsilon)$ such that 
$\nu^*_n\xi_T \in [s_i^n,s_j^n]$. Then $f_n(\nu^*_n\xi_T)\in [f_n(s_i^n),f_n(s_j^n)]=[f(s_i^n),f(t_i^n)]$. Thus, $$f_n(\nu^*_n\xi_T)\in(f(\nu^*\xi_T)-\epsilon,f(\nu^*\xi_T)+\epsilon), $$ i.e.,  the arbitrariness of $\epsilon>0$ yields $f_n(\nu^*_n\xi_T) \to f(\nu^*\xi_T)$ a.s..  Using DCT, we obtain the first result. \\
Second, we prove that 
 $X_{n,t}^* \overset{L^r(\Omega)}{\longrightarrow} X_t^*$  a.e. $t\in[0,T)$ and $X_{n}^* \xrightarrow{L^r(\Omega \times [0,T])}  X^*$ hold. \\
Using Lemma \ref{lemma:L^2-co}, we have 
$$\lim\limits_{n\to \infty}\E\left[(\xi_Tf(\nu^*\xi_T)-\xi_Tf_n(\nu^*_n\xi_T))^2\right] =0.$$ 
Using Cauchy-Schwarz's inequality yields
$$
\E\left[\left|X_{n,T}^*-X_T^*\right|\right]=\E\left[\left|f_n(\nu_n^*\xi_T)-f(\nu^*\xi_T)\right|\right] \leq \(\E\left[(\xi_Tf_n(\nu^*_n\xi_T)-\xi_Tf(\nu^*\xi_T))^2\right]\)^{\frac{1}{2}}\cdot \(\E\left[\xi_T^{-2}\right]\)^{\frac{1}{2}},
$$ 
which indicates  $X_{n,T}^* \overset{L^1(\Omega)}{\longrightarrow} X_T^*$. Moreover, as $X_t^*=\xi_t^{-1}\E\left[\xi_Tf(\nu^*\xi_T)|\F_t\right]$, we have 
\begin{align*}\small
&\mathbb{E}\left[ \left| \xi_t^{-1}\mathbb{E}[\xi_TX_{T}^* | \mathcal{F}_t]  - \xi_t^{-1}\mathbb{E}[\xi_TX_{n,T}^*| \mathcal{F}_t] \right| \right] 
\leq \mathbb{E}\left[  \mathbb{E}\left[\frac{\xi_T}{\xi_t}\left|X_{T}^* -X_{n,T}^*\right| \Big{|} \mathcal{F}_t\right]\right] \\
&= \mathbb{E}\left[ \frac{\xi_T}{\xi_t}\left|X_{T}^* - X_{n,T}^*\right| \right]\leq\E\left[\xi_t^{-2}\right]^{\frac{1}{2}} \cdot \E\left[|\xi_T^2 (X_T^*-X_{n,T}^*)|^2\right]^{\frac{1}{2}},
\end{align*}
which implies $X_{n,t}^* \overset{L^1(\Omega)}{\longrightarrow} X_t^*$ for any $t\in[0,T)$. In addition, using Fubini's theorem, we obtain
$$\E\left[\int_0^T\left|X_{n,t}^*-X_t^*\right| \d t\right]\leq 
 \(\E\left[|\xi_T^2(X_T^*-X_{n,T}^*)|^2\right]\)^{\frac12}\cdot \int_0^T \(\E\left[\xi_t^{-2}\right]\)^{\frac12} \d t,
$$
which implies $X_{n}^* \xrightarrow{L^1(\Omega \times [0,T])}  X^*$. \\
Third, we prove $\boldsymbol{\pi}_{n}^{*} \xrightarrow{L^1(\Omega \times [0,T])}  \boldsymbol{\pi}^{*}$.  \\  According to the martingale representation theorem, we have that there exists a square integrable random process $\boldsymbol{\psi}_t$ such that $\E\left[\xi_Tf(\nu^*\xi_T)\right|\F_t]=\int_0^t \boldsymbol{\psi}_s \d \mathbf{W}_s$, and we can similarly define $\{\boldsymbol{\psi}_t^n\}_{\{t\in [0,T]\}}$. Using Itô's isometry, we obtain
$$\E\left[\int_0^T ||\boldsymbol{\psi}_t^n-\boldsymbol{\psi}_t||^2_2 \d t\right]=\E\left[(\xi_TY_T^n-\xi_TY_T)^2\right]=\E\left[(\xi_Tf_n(\nu^*_n\xi_T)-\xi_Tf(\nu^*\xi_T))^2\right].$$
Using Fubini's theorem and Cauchy-Schwarz's inequality, 
\begin{align*}\small
\E\left[\int_0^T \left|\left|\xi_t^{-1}\boldsymbol{\sigma}^{-1}(\boldsymbol{\psi}^n_t-\boldsymbol{\psi}_t)\right|\right|_2 \d t\right]
&\leq\left|\left|\boldsymbol{\sigma}^{-1}\right|\right|_2\cdot\left(\int_0^T\E\left[\xi_t^{-2}\right] \d t\right)^{\frac12}\cdot\left(\int_0^T\E\left[\left|\left|\boldsymbol{\psi}^n_t-\boldsymbol{\psi}_t\right|\right |^2_2\right] \d t\right)^{\frac12}.
\end{align*}
Thus, $$\lim\limits_{n\to\infty}\E\left[\int_0^T \left|\left|\xi_t^{-1}\boldsymbol{\sigma}^{-1}(\boldsymbol{\psi}^n_t-\boldsymbol{\psi}_t)\right|\right|_2\d t\right] 
= 0.$$
Using Fubini's Theorem and Jensen's inequality,  
\begin{align*}\small
\E\left[\int_0^T \left|\left|\left(Y_t^n-Y_t\right)\boldsymbol{\sigma}^{-1}\boldsymbol{\theta}\right|\right|_2 \d t\right]
&\leq \left|\left|\boldsymbol{\sigma}^{-1}\boldsymbol{\theta}\right|\right|_2 \cdot \int_0^T\mathbb{E}\left[\E\left[ \frac{\xi_T}{\xi_t}\left|f(\nu^*\xi_T) - f_n(\nu^*_n\xi_T)\right|\Big|\F_t\right] \right]\d t
\\
&\leq \left|\left|\boldsymbol{\sigma}^{-1}\boldsymbol{\theta}\right|\right|_2\cdot\(\E\left[(\xi_TY_T^n-\xi_TY_T)^2\right]\)^{\frac12}\cdot\int_0^T \(\E\left[\xi_t^{-2}\right]\)^{\frac{1}{2}}\d t.
\end{align*}
Thus, according to Theorem 6.3 in \cite*{KLSX1991}, we immediately obtain
$$\E\left[\int_0^T \left|\left|\boldsymbol{\pi}_{n,t}^*-\boldsymbol{\pi}_t^*\right|\right|_2\d t\right]=\E\left[\int_0^T\left|\left|\xi_t^{-1}\boldsymbol{\sigma}^{-1}(\boldsymbol{\psi}^n_t-\boldsymbol{\psi}_t)+(Y_t^n-Y_t)\boldsymbol{\sigma}^{-1}\boldsymbol{\theta}\right|\right|_2\d t\right] \to 0\ ( n\to\infty).$$
Fourth, we prove \textbf{ $X_{n,t}^* \overset{\text{a.s.}}{\longrightarrow} X_t^*$ \ a.s. \ $t\in[0,T]$.} 

We have $X_t^*=\xi_t^{-1}\E\left[\xi_Tf(\nu^*\xi_T)|\F_t\right]=\int_0^\infty xf(\nu^*\xi_tx)\d F_t(x)$. By Lemma \ref{prop:1}, for any fixed $t\in[0,T)$, and $\xi_t\in(0,\infty)$, we have $\left|\int_{0}^{\infty}\left[xf(\nu^*\xi_tx)\right]^{\frac43}\d F_{t}(x)\right|$ has a uniform upper bound. Therefore, we obtain $X_{n,t}^* \to X_t^*$ a.s. based on the properties of uniform integrability and the fact that $f_n(\nu^*_n\xi_tx) \to f(\nu^*\xi_tx)$ a.e. for a fix $\xi_t$. \\
At last, we prove  $\boldsymbol{\pi}_{n,t}^* \overset{L^r(\Omega)}{\underset{\text{a.s.}}{\longrightarrow}} \boldsymbol{\pi}_t^*$ \ a.s.\  $t\in[0,T)$.
\\
In order to verify the a.s. convergence, we only need to prove $\int_0^\infty f_n(\nu^*_n\xi_tx)(xK(x))'\d x$ converge to $\int_0^\infty f(\nu^*\xi_tx)(xK(x))'\d x$ as $n \to \infty$ based on the proof of Lemma \ref{lemma:twice-dif}. Imitating the proof of the fourth step, we easily conclude the result. As for the $L^1(\Omega)$ convergence, using
\begin{align*}\small
&\int_0^\infty f_n(\nu^*_n\xi_tx)(xK(x))'\d x=\int_0^\infty f_n(\nu^*_n\xi_tx)\left[x+\sum_{i\leq\tau_1}c_i^{n}x^{\alpha_i^n+2}(\ln x)^{\beta_i^n}\right] \d F_{t}(x)\\&=\E\left[\frac{\xi_T}{\xi_t}f_n(\nu^*_n\xi_T)\Big|\F _t\right]+\sum_{i\leq \tau_1}c_i^n\E\left[\left(\frac{\xi_T}{\xi_t}\right)^{\alpha_i^n+2}\left(\ln \frac{\xi_T}{\xi_t}\right)^{\beta_i^n}f_n(\nu^*_n\xi_T)\Big| \F_t\right],
\end{align*}
and the result of Lemma \ref{prop:1}, and imitating the proof of the second step, we have that each of the above terms converges in $L^1(\Omega)$. Thus the final result follows.
\end{proof}
\begin{remark}\label{remark:convergence}
From the fact that $X_{n,t}^* \overset{L^r(\Omega)}{\underset{\text{a.s.}}{\longrightarrow}} X_t^*$ and $\pi_{n,t}^{*,d} \overset{L^r(\Omega)}{\underset{\text{a.s.}}{\longrightarrow}}\pi_t^{*,d}$ hold for any $t\in [0,T]$, it can be seen that the optimal wealth process and the optimal portfolio exhibit convergence of random variables at fixed time points. This indicates that, the wealth generated by the PHARA-approximated portfolio can be arbitrarily close to that of the optimal portfolio at each time point. From $X_{n}^* \xrightarrow{L^r(\Omega \times [0,T])}  X^*$ and $\pi_{n}^{*,d} \xrightarrow{L^r(\Omega \times [0,T])}  \pi^{*,d}$, we observe that they satisfy the convergence of trajectories as a whole. This suggests that, throughout the entire investment process, the expectation of the cumulative error resulting from the PHARA-approximated portfolio can be made arbitrarily small.
\end{remark}
\subsection{Convergences of the Preference-fitting Method and PHARA Approximation}\label{subsec:fittingrate}
We return to the preference-fitting method introduced in Section \ref{sec:fittingmethod}, establishing and summarizing its convergences. Assume the optimal wealth process of the investor is given by $\{X_t^*\}_{0\le t\le T}$ and the optimal portfolio is given by $\{\boldsymbol{\pi_t}\}_{0\leq t\leq T}$.

If $\max\limits_{i=0,1,\cdots, n-1} \{|\xi_i^n-\xi_{i+1}^n| \} \to 0$, as $n\to \infty$, it is easy to verify that $X_T^n(\xi^-)$ and $X_T^n(\xi^+)$ exist for any $\xi>0$ due to the monotonicity of $\{y_i^n\}_{0\leq i\leq n}$, and thus $X_T^n$ converge a.s. (to a certain $X_T^*$) due to the fact that the discontinuous points of $X_T^*$ form a countable set. In addition, we have $$\E\left[(\xi_TX_T^n-\xi_TX_T^*)^2\right]\leq 2y_1^1\cdot\E[\xi_T]\cdot\E[\left|\xi_TX_T^n-\xi_TX_T^*\right|].$$ Thus, we conclude $X_T^n\to X_T^*$ in  $L^r(\Omega,\F, \p)$, $r\in[1,2)$ and $\E\left[(\xi_TX_T^n-\xi_TX_T^*)^2\right] \to 0$ as $n\to \infty$. Because $X_T^n\in \mathcal{C}$, using a similar way to Theorem \ref{th:co-X}, we obtain the convergence in the sense of a.s., $L^r$ ($r\in[1,2)$), and $L^1(\Omega \times [0,T])$ for the fitted wealth process and the fitted portfolio.

For fixed $t\in[0,T]$, when analyzing uniform convergence with respect to $\xi_t$ (market conditions), we treat the random variables as real-valued functions of $\xi_t$, with the convergences pertaining to the real-valued functions defined on $(0,\infty)$. We have the following results.
\begin{theorem}\label{th:fitcon}
Assume $X_t^n$ and $\boldsymbol{\pi}_t^n$ are given in Definition \ref{def:fitXpi}.
If $\max\limits_{\{i=0,1,\cdots, n-1\}} \{|\xi_i^n-\xi_{i+1}^n| \} \to 0$ as $n\to \infty$, we have the following results.\\
(1) $X_t^n\to X_t^*$ and $\boldsymbol{\pi}_t^n\to \boldsymbol{\pi}_t^* $ a.e. uniformly on $(0,\infty)$ for each $t\in [0,T)$.\\
(2) In addition, if $\X_T^*$ is continuous on $(0,\infty)$, we have $X_t^n\to X_t^*$ and $\boldsymbol{\pi}_t^n\to \boldsymbol{\pi}_t^* $ a.e. uniformly on $ (0,\infty)$ for each $t\in[0,T]$.
\end{theorem}
\begin{proof}
(1) By the definition of $X_T^n$, we have
\begin{align*}
X_t^n-X_t^*&=\xi_t^{-1}\E[\xi_T\X_T^n(\xi_T)-\xi_T\X_T^*|\F_t]=\int_0^\infty [\X_T^n(\xi_t x)-\X_T^*(\xi_tx)]K(x,t)\d x\\&=\int_{\xi^{p_0^1}}^{\xi^{p_1^1}}\frac{1}{\xi_t^2}[\X_T^n(x)-\X_T^*(x)]K\Big(\frac{x}{\xi_t},t\Big)\d x.
\end{align*}
Direct calculation indicates that $\Big\{\Big|\frac{1}{\xi_t^2}K\Big(\frac{x}{\xi_t},t\Big)\Big|\Big\}_{\xi_t>0}$ is bounded by a constant $C>0$. Hence, we have
$$|X_t^n-X_t^*|\leq C \int_{\xi^{p_0^1}}^{\xi^{p_1^1}}|\X_T^n(x)-\X_T^*(x)|\d x\to 0\quad (n\to \infty).$$
(2) If $\X_T^*$ is continuous, we have
\begin{align*}
|X_t^n-X_t^*|&\leq \xi^{p_1^1}\cdot w\Big(\max\limits_{\{i=0,1,\cdots, n-1\}} \{|\xi_i^n-\xi_{i+1}^n| \}\Big)\cdot \int_0^\infty \frac1x K(x,t)\d x\\&=\xi^{p_1^1}\cdot w\Big(\max\limits_{\{i=0,1,\cdots, n-1\}} \{|\xi_i^n-\xi_{i+1}^n| \}\Big) \to 0 \quad (n\to \infty),
\end{align*}
where $w$ is the modulus of continuity for $\X_T^*$ on $[\xi^{p_0^1},\xi^{p_1^1}]$.

To show the uniform convergence of the fitted portfolio, we just need to prove the uniform convergence of $\int_0^\infty \X_n(\xi_tx)\frac{\partial(xK(x,t))}{\partial x}\d x$, according to Theorem \ref{th:co-X}. This is similar to the proof for $X_t^n$.
\end{proof}
In the case where $X_T^*$ is discontinuous (there exists linear parts in the elicited utility function), it is easy to see that the uniform convergence of $X_T^n$ fails at the discontinuity points.
While according to Theorem \ref{th:fitcon}, as $K$ is smooth, if the time $t<T$, the convergence remains uniform. Moreover, if the terminal wealth function $\X$ is continuous, the convergence is uniform over the entire investment horizon.
\begin{remark}
Theorem \ref{th:fitcon} shows that, for any fixed time instant, the convergence is effective regardless of market conditions. Additionally, the proof of (2) suggests a stronger result: if $\X$ is continuous, then $X_t^n \to X_t^*$ and $\boldsymbol{\pi}_t^n \to \boldsymbol{\pi}_t^*$ a.e. uniformly with respect to $(t,\xi_t) \in [0,T] \times (0,\infty)$.
\end{remark}
Using Theorems \ref{th:co-X}-\ref{th:fitcon}, we immediately obtain the following corollary regarding convergence rates.
\begin{corollary}
Assume $\X_T^*\in C^1(0,\infty)$ and $\max\limits_{\{i=0,1,\cdots, n-1\}} \{|\xi_i^n-\xi_{i+1}^n| \}= O(\frac1n)$. The convergence rates are given by $|X_t^n-X_t^*|=O(\frac{1}{n})$, a.s., $\E[|X_t^n-X_t^*|]=O(\frac1n)$ and $\E\left[\int_0^T|X_t^n-X_t^*|\d t\right]=O(\frac1n)$ for any $t\in [0,T]$.
\end{corollary}
The various convergences for the preference-fitting method are summarized in Table \ref{tab:convergence}.
\begin{table}[h!]\small
    \centering
    \begin{tabular}{|c|c|c|c|c|c|c|}
        \hline
                & a.s. & $L^r(\Omega),r\in[1,2)$ & $L^1(\Omega\times [0,T])$  & uniform & locally uniform ($t<T$) \\ \hline
        with linear parts in $V$    & \checkmark    & \checkmark    & \checkmark        &     &\checkmark     \\ \hline
        without linear part    & \checkmark    & \checkmark    & \checkmark        & \checkmark    & \checkmark    \\ \hline
    \end{tabular}
    \caption{Convergences of $X_t^n$ and $\pi_{t}^{n,d}$ for the preference-fitting method. The uniform convergence is with respect to the real-valued functions of $(t,\xi_t)$, and the locally uniform convergence is with respect to $\xi_t$ for fixed $t$.}
    \label{tab:convergence}
\end{table}

\textbf{Step 1} introduced in the Preference-Fitting method, namely, specifying preferences under extreme market conditions, is a technical device that enables us to easily obtain stronger uniform convergence results. In contrast, for the PHARA approximation approach, the analysis of uniformity becomes more complicated. First, $0$ and $\infty$ corresponding to extreme market conditions might become two ``discontinuity points", which destroys the uniform convergence (the example in Subsection \ref{sec:ex_hyper} reflects the influence of this when $X_t^*$ lies in the first segment. See also Figure \ref{table:Td}). Second, the Lagrange multiplier is no longer fixed, but appears as a convergent sequence $\{\nu_n^*\}_{n\geq 1}$, whose error may affect the uniformity. Concerning the PHARA approximation, some conclusions on uniform convergence are as follows.

For PHARA approximated wealth process and portfolio, the uniform convergence holds with respect to $\xi_t\in [\epsilon,\frac1\epsilon]$ for each fixed $t\in [0,T)$ (and for $t=T$ if $f=T_L(V)$ is continuous). We provide a outline of proof. Note
\begin{align*}
|X_{n,t}^*-X_t^*|&=\Big|\int_0^\infty [f_n(\nu_n^*\xi_t x)-f(\nu^*\xi_t x)]K(x,t)\d x\Big|\\&=\Big|\int_0^\infty \frac{1}{\nu^*_n\xi_t}(f_n(x)-f(x))K\Big(\frac{x}{\nu_n^*\xi_t},t\Big)\d x+\int_0^\infty K\Big(x,t\Big)\Big(f(\nu_n^*\xi_tx)-f(\nu^*\xi_tx)\Big)\d x\Big|.
\end{align*}
For any fixed $t\in[0,T)$, there exists $\tilde{K}:\R^+\times [0,T]\to \R^+$ such that $K\Big(\frac{x}{\nu_n^*\xi_t},t\Big)\leq \tilde{K}(x,t)$ and $\int_0^\infty (f_n(x)-f(x))\tilde{K}\Big(x,t\Big)\d x\to 0$, $n\to \infty$ for any $x\geq 0$ and $\xi_t \in [\epsilon,\frac1\epsilon]$.
The function $\tilde{K}$ can be taken as
$$\tilde{K}(x,t)=\begin{cases}
K(\frac{x}{\nu_{inf}\epsilon},t), &0\leq x<x_1,\\
C_1,& x_1\leq x\leq x_2 ,\\
K(\frac{\epsilon x}{\nu_{sup}},t),&x_2<x,
\end{cases}$$
where $C_1$, $x_1$ and $x_2$ are constants with regarding $t$, $\epsilon$, $\nu_{inf}$ and $\nu_{sup}$.
Additionally, we can similarly prove $\Big|\int_0^\infty \Big(K\Big(\frac{x}{\nu_n^*\xi_t},t\Big)-K\Big(\frac{x}{\nu^*\xi_t},t\Big)\Big)f(x)\d x\Big |\leq a_n\to 0$, where $\{a_n\}_{n\geq 1}$ is independent of $\xi_t$. Thus, we obtain the uniform convergence. The above result for PHARA approximation approach suggests that when market conditions are particularly good or particularly bad, we need to increase the accuracy of the PHARA approximation to ensure a good fitting performance.

Moreover, for any $\epsilon>0$, there exists $I\subset [0,T]\times (0,\infty)$ such that $\p((t,\xi_t)\in I)>1-\epsilon$ and the uniform convergence holds for $\{X_{n,t}^*\}_{n\geq 1}$ and $\{\boldsymbol{\pi}_{n,t}^*\}_{n\geq 1}$ with respect to $(t,\xi_t)\in I$. This is a direct corollary of Egoroff's theorem, indicating that after excluding a small-probability event, the method enjoys uniform convergence. We list the convergences for PHARA-approximation in Table \ref{tab:convergence2} below.
\begin{table}[h!]\small
    \centering
    \begin{tabular}{|c|c|c|c|c|c|c|}
        \hline
               a.s. & $L^r(\Omega),r\in[1,2)$ & $L^1(\Omega\times [0,T])$  & uniform & locally uniform ($\xi_t\in[\epsilon,\frac1\epsilon]$) & Egoroff convergence\\ \hline
             \checkmark    & \checkmark    & \checkmark        &     &\checkmark  & \checkmark   \\ \hline
    \end{tabular}
    \caption{Convergences for the PHARA-approximation and the preferece-fitting method without specifying extreme market conditions. The uniform convergence is with respect to $\xi_t$ for fixed $t$.}
    \label{tab:convergence2}
\end{table}
\subsection{The Explicit Expression for PHARA Approximation}\label{subsec:expl}
Having an explicit expression is an advantage of the PHARA approximation. Recalling Definition \ref{def:V_n}, we additionally define $s_{M(n)+1}^n=T(n)$, $s^n_{M(n)+2}=t_{M(n)+1}^n=\infty$, $\tau_{n,i}^{-}=(V_n)'_-(s_i^n)$, $\tau_{n,i}^+=(V_n)'_+(s_i^n)$ for $i=0,1,\cdots,M(n)+1$. We demonstrate the explicit expressions for the PHARA-approximated wealth and the PHARA-approximated portfolios in Theorems \ref{th:phara_results} and \ref{th:log_phara}, the proofs of which are similar to the proofs of Proposition 2 and Theorem 1 in \cite*{LLMV2024}.

\begin{theorem}\label{th:phara_results}
For the $n$-th division,  we have\\
(1) The PHARA-approximated terminal wealth is given by
\begin{equation}\small
\begin{aligned}
X_{n,T}^*=&\sum_{k=1}^{M(n)+1}\Bigg\{s_{k}^n\id_{\left\{\nu^*_n\xi_T\in \left(\tau_{n,k}^+,\tau_{n,k}^-\right)\right\}}
+\left(u_k^n+\left(\frac{\tau_{n,k}^+}{\nu^*_n\xi_T}\right)^{\frac{1}{1-\gamma_k^n}}\left(s_{k}^n-u_k^n\right)\right)\id_{\left\{\nu^*_n\xi_T\in\left(\tau_{n,k+1}^-,\tau_{n,k}^+\right)\right\}}\times \id_{\left\{\gamma_k^n< 1\right\}}
\Bigg\}.
\end{aligned}
\end{equation}
(2) The PHARA-approximated wealth process at $t\in [0,T)$ is given by
\begin{equation}
\begin{aligned}
X_{n,t}^*:=X_t^D+X_t^A+X_t^R=\sum_{k=1}^{M(n)+1}(X_{t,k}^D+X_{t,k}^A+X_{t,k}^R),
\end{aligned}
\end{equation}
where \begin{small}
\begin{align*}
&X_{t,k}^D=e^{-r(T-t)}s_k^n\left[\Phi\left(d_1\left(\frac{\tau_{n,k}^+}{\nu^*\xi_t}\right)\right)-\Phi\left(d_1\left(\frac{\tau_{n,k}^-}{\nu^*\xi_t}\right)\right)\right],\quad \\&
X_{t,k}^A=e^{-r(T-t)}u_k^n\left[\Phi\left(d_1\left(\frac{\tau_{n,k+1}^+}{\nu^*\xi_t}\right)\right)-\Phi\left(d_1\left(\frac{\tau_{n,k}^+}{\nu^*\xi_t}\right)\right)\right]\times \id_{\left\{\gamma_k^n < 1\right\}}, \\&
X_{t,k}^R=e^{-r(T-t)}(s_k^n-u_k^n)\frac{\Phi'\left(d_1\left(\frac{\tau_{n,k}^+}{\nu^*\xi_t}\right)\right)}{d^{\gamma_k^n}\left(\frac{\tau_{n,k}^+}{\nu^*\xi_t}\right)} 
\left[\Phi\left(d^{\gamma_k^n}\left(\frac{\tau_{n,k+1}^-}{\nu^*\xi_t}\right)\right)-\Phi\left(d^{\gamma_k^n}\left(\frac{\tau_{n,k}^+}{\nu^*\xi_t}\right)\right)\right]\times \id_{\left\{\gamma_k^n < 1\right\}},
\end{align*}
\end{small}
and
\begin{align}\small
    d_1(z)
    :=\frac{1}{-||\boldsymbol{\theta}||_2\sqrt{T-t}}\left(\log(z)+\left(r-\frac{{||\boldsymbol{\theta}||}_2^2}{2}\right)(T-t)\right),
    \label{eq_d1}
\quad d^{\gamma_k^n}(z)&:=d_1(z)-\frac{||\boldsymbol{\theta}||_2\sqrt{T-t}}{1-\gamma_k^n},z>0.
\end{align}
(3) The PHARA-approximated portfolio at time $[0,T)$ is given by
\begin{equation}
\begin{aligned}
\boldsymbol{\pi}_{n,t}^*=(\boldsymbol{\sigma}^\top)^{-1}\boldsymbol{\theta}\sum_{k=0}^{M(n)+1}\Bigg\{\frac{1}{1-\gamma_k^n}X_{t,k}^R+ e^{-r(T-t)}\frac{s_{k+1}^n-s_k^n}{||\boldsymbol{\theta}||\sqrt{T-t}}\Phi'\left(d_1\left(\frac{\tau_{n,k}^+}{\nu_n^*\xi_t}\right)\right)\id_{\left\{\gamma_k^n=1\right\}}
\Bigg\}.
\end{aligned}
\end{equation}
\end{theorem}
In the PHARA approximation method, we sometimes only use the logarithmic base for convenience in calculations, and thus we have the following more specific result.
\begin{theorem}\label{th:log_phara}
If  $\gamma_i^n=0$ or $1$, then the PHARA-approximated portfolio is given by
\begin{equation}\label{eq:pi}
\small\begin{aligned}
\boldsymbol{\pi}_{n,t}^*&=\underbrace{(\boldsymbol{\sigma}^{\top})^{-1}\boldsymbol{\theta}X_{n,t}^*}_{\text{Merton term}}\quad+\underbrace{\frac{e^{-r(T-t)}}{\sqrt{T-t}}\cdot\frac{(\boldsymbol{\sigma}^\top)^{-1}\boldsymbol{\theta}}{||\boldsymbol{\theta}||}\sum_{k=0}^{M(n)+1}(s_{k+1}^n-s_k^n)\Phi'\left(d_1\left(\frac{\tau_{n,k}^+}{\nu_n^*\xi_t}\right)\right)\id_{\left\{\gamma_k^n=1\right\}}}_{\text{risk seeking}}\\&\quad -\underbrace{e^{-r(T-t)}(\boldsymbol{\sigma^\top})^{-1}\boldsymbol{\theta}\sum_{k=0}^{M(n)+1}u_k^nq_k^n\id_{\left\{\gamma_i^n=0\right\}}}_{\text{loss aversion}}-\underbrace{e^{-r(T-t)}(\boldsymbol{\sigma}^\top)^{-1}\boldsymbol{\theta}\sum_{k=0}^{M(n)+1}s_k^np_k^n}_{\text{first-order risk aversion}}\\&:=\boldsymbol{\pi}_t^{(1)}+\boldsymbol{\pi}_t^{(2)}+\boldsymbol{\pi}_t^{(3)}+\boldsymbol{\pi}_t^{(4)},
\end{aligned}
\end{equation}
where 
\begin{equation}\small
\begin{aligned}
p_k^n:=\Phi\left(d_1\left(\frac{\tau_{n,k}^+}{\nu_n^*\xi_t}\right)\right)-\Phi\left(d_1\left(\frac{\tau_{n,k}^-}{\nu_n^*\xi_t}\right)\right), \quad 
&q_k^n:=\Phi\left(d_1\left(\frac{\tau_{n,k+1}^-}{\nu_n^*\xi_t}\right)\right)-\Phi\left(d_1\left(\frac{\tau_{n,k}^+}{\nu_n^*\xi_t}\right)\right).
\end{aligned}
\end{equation}
\end{theorem}
As $\sum_{k=0}^np_k^n+\sum_{k=0}^nq_k^n=1$, $p_k$ and $q_k$ are interpreted to be probabilities. The magnitude of the probabilities intuitively reflects the weights of the different terms in Eq. (\ref{eq:pi}) in the investor's portfolio.

Similar to \cite*{LLMV2024}, the first term $\boldsymbol{\pi}_t^{(1)}$ represents a constant percentage portfolio which is called the Merton relative risk aversion term. When the investor employs a logarithmic utility function or a power utility function, the optimal portfolio for Problem (\ref{eq:prob_intr}) derived from the martingale-duality method is exactly $\boldsymbol{\pi}_t^{(1)}$. The second term $\boldsymbol{\pi}_t^{(2)}$ is defined as the risk-seeking term, which is closely associated with the linear parts in the utility function. We observe that $\boldsymbol{\pi}_t^{(2)}$ leads to a significantly increase in the allocation to risky assets. The third term $\boldsymbol{\pi}_t^{(3)}$ represents the loss-aversion term, which induces a decrease of the risky investment to avoid loss. The fourth term $\boldsymbol{\pi}_t^{(4)}$ is called the first-order risk aversion term. It arises from the non-differentiable points of the utility function.

The explicit forms of the fitted wealth process and the fitted portfolio introduced in Section \ref{sec:fittingmethod} can be derived in the same way and thus we omit the proof.

In Section \ref{sec:fittingmethod}, as we use the hyperbolic expressions to link the points derived from the probability-wealth pairs, the utility function corresponding to $X_T^k$ is a PHARA utility function whose specific form can be obtained using the bijection established in Section \ref{sec:bijection}. We can obtain the explicit form of the fitted wealth process and fitted portfolio in analogy to Theorem \ref{th:phara_results}.

\section{Application and Numerical Illustration}
\subsection{A Numerical Simulation of the PHARA Approximation}\label{sec:ex_hyper}
We present specific example of the PHARA approximation in this subsection.
We refer to the utility function with the form $U(x)=\ln x -\frac{\lambda}{x},\lambda\geq 0$ as the log-hyperbolic utility function. This utility function incorporates a penalty term into the conventional logarithmic utility to mitigate the risk of extremely low wealth. 
We seek to approximate the optimal portfolio using PHARA utilities to obtain an explicit formulation. We assume $\lambda = 1$ and present its PHARA approximation as follows.

The utility function is given by $V(x)  = \ln x - \frac{1}{x}$. We employ logarithmic bases to formulate each component of $V_n$, and let $T(n) \in \mathbb{N}^*$. Furthermore, we define the length of each segment within the interval $[0, T(n)]$ as $d_n = \frac{1}{m}$, where $m \in \mathbb{N}^*$. Specifically, for the $n$-th division, we have $$V_n(x)=\sum_{i=1}^{\frac{T(n)}{d_n}}\left[a_i^n\ln (x-u_i^n)+b_i^n\right]\id_{\{x\in [(i-1)d_n,id_n)\}}+\left[\ln \left(x-u_{T(n)+1}^{n}\right)+b_{T(n)+1}^n\right]\id_{\{x\in[T(n),\infty)\}}.$$ By straightforward calculation, we have $u_i^n=\frac{i(i-1)d_n}{i(i-1)d_n+(2i-1)}$ and $a_i^n=1+\frac{1}{i(i-1)d_n^2+(2i-1)d_n}$ for $n\leq M(n)$ and $u_{T(n)+1}^n=\frac{T(n)}{T(n)+1}$. Then, the explicit expressions for the PHARA approximated wealth and portfolio can be derived according to Theorem \ref{th:phara_results}.

We let $m=1$, $r=0.05$, $T-t=1$, $\sigma=0.5$, $\theta=0.25$ and $x_0=2$. For various values of $T_d$ and $d$, the values of $\nu^*$ are shown in Table \ref{table:Td}. 
We examine the approximating effect on the position ratio, expressed as $\frac{\boldsymbol{\pi}_t}{X^*_t}$, as illustrated in Figure \ref{fig:fitting}. The observed convergence of the position ratio is not uniform, prompting our investigation into locally uniform convergence. As $X_t^*\to 0$ and $X_t^* \to \infty$, the limit of the original position ratio is 0.25 and 0.5, respectively. As $X_t^*$ nears 0, the predominant term in the utility function influencing the portfolio is $-\frac{1}{x}$, resulting in the original position ratio being close to 0.25. However, the PHARA approximation employs a logarithmic base for each segment, which reflects the characteristics of logarithmic utility and leads to the approximated position ratio close to 0.5 when $X_t^*$ nears 0. This discrepancy accounts for the failure of uniform convergence. Furthermore, as both $\frac{1}{X_{n,t}^*}$ and $\pi_{n,t}^*$ converge locally uniformly and exhibit uniform bounds with respect to $\xi_t \in I \subset (0, \infty)$, where $I$ is compact, the locally uniform convergence holds for the position ratio, consistent with the results in Table \ref{tab:convergence2}.
\begin{table}[ht]
\begin{varwidth}[b]{0.4\linewidth}\large
    \centering
    \begin{tabular}{c|c|c|c}
        & $T_d=1$ & $T_d=5$ & $T_d=10$ \\
        \hline
        $d=0.2$ & 0.811 & 0.735 & 0.735 \\
        \hline
        $d=0.1$ & 0.801 & 0.735 & 0.735 \\
        \hline
        $d=0.01$ & 0.793 & 0.735 & 0.735 \\
        \hline
    \end{tabular}
    \quad\\[50pt]
    \caption{The values of $\nu^*$.}
    \label{table:Td}
\end{varwidth}
\hfill
\begin{minipage}[b]{0.5\linewidth}
    \centering
    \includegraphics[width=\linewidth]{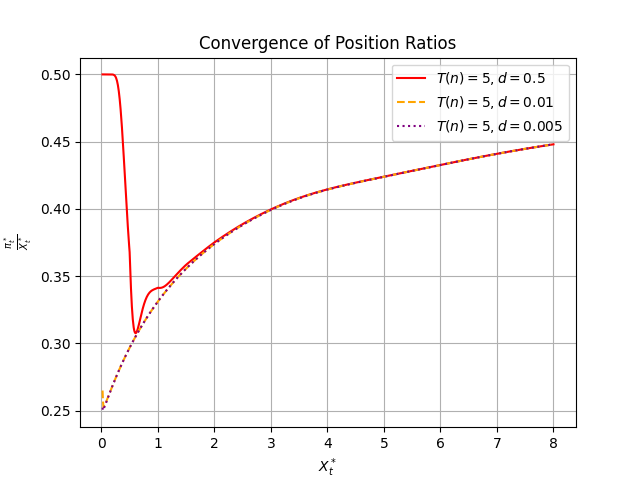}
    \captionof{figure}{The position ratio.}
    \label{fig:fitting}
\end{minipage}
\end{table}


\subsection{Comparison with the Optimal Problem under VaR Constraints}\label{sec:VaR}
In this section, we compare the preference-fitting method with the optimization problem under VaR constraints proposed by \cite*{BS2001}. We first consider the traditional framework and assume that the manager has an S-shaped utility introduced in \cite*{kahneman1979d}:
\begin{equation}\label{fuc_S-shape}
    \hat{U}(x) = \left\{
    \begin{aligned}
    & k_1(x-B_{1})^{\gamma_1} && x \geq B_1,\\
    & - k_2(B_{1}-x)^{\gamma_2} && x < B_1,
    \end{aligned}
    \right.
\end{equation}
where $B_1$ is the reference point, and $k_1>0$, $k_2>0$, $0<\gamma_1< 1$ and $0<\gamma_2 <1$ measure the degree of risk aversion and risk seeking. The manager compares her payoff with the reference level $B_1$. She becomes loss-averse (risk-averse) below (over) the reference level.
From the perspective of portfolio-induced utility, when an investor introduces a VaR constraint to enhance the stability of the portfolio, she is essentially altering her utility function. We conduct the following analysis to explain this viewpoint.

We use the same notation as in Section \ref{sec:model}. In this section, we formulate a constrained utility optimization problem with a deterministic benchmark:
\begin{equation}\label{prob-VaR}
	\begin{aligned}
		& \max_{\boldsymbol{\pi} \in \Pi} \E[ \hat{U}(X_T^{\boldsymbol{\pi}})]\quad
  \text{subject to } \p(X_T \geq L)\geq 1-\alpha,
	\end{aligned}
\end{equation}
where $\p(X_T\geq L)\geq 1-\alpha$ is the VaR constraint.

We consider the following auxiliary problem:
\begin{equation}\label{prob-auxiliary}
	\begin{aligned}
		& \max_{Z \in \mathcal{M}} \E[ \hat{U}(Z)]\quad 
  &
  \text{subject to } \E[\xi_TZ] \leq x_0,  \ \mbox{and} \ \p(Z \geq L)\geq 1-\alpha,
	\end{aligned}
\end{equation}
where $\mathcal{M}$ is the set of all the random variables $Z$  satisfying $Z+C e^{rT}\geq0$ a.s. for some $C\geq 0$.

According to \cite*{dong2020optimal}, for $\nu_1>0$, we define the modified utility function
\begin{equation}\label{eq:U_nu}
U^{\nu_1}(x):= 
    \begin{cases} 
        \hat{U}(x)+\nu_1\id_{\{x \geq L\}}\  & x \geq 0, \\
        -\infty\  & x < 0.
    \end{cases}
\end{equation}
In the existing literature, the VaR constraint is used widely as a constraint in the optimal portfolio selection to manage risk. \cite*{BS2001} first embed this concept into the portfolio selection problem. They exclusively addresses a special case for utility-based investors. In fact, in their framework, the problem reduces to that of an investor maximizing a transformed utility function Eq. (\ref{eq:U_nu}) without the VaR constraints, where the multiplier $\nu_1$ is determined by the complementary slackness in Eq. (\ref{prob-auxiliary}). In general, due to the non-convexity of the feasible portfolio set, the complementary slackness may not hold, leading to the non-existence of Lagrange multipliers and the investor may no longer be utility-based; see footnote 5 in \cite*{BS2001}. The rigorous characterization of the optimal portfolio in this setting has not been formally established. We then provide some observations regarding the solution procedure for the optimization problem involving VaR constraints and draw a comparison with the preference-fitting method as outlined below.

Under the preference-fitting framework, the investor can focus solely on the probability-wealth pair $(1-\alpha, y)$ corresponding to the VaR constraint, integrating it into the existing system of the upper and lower budget bounds. Because the fitting procedure applies to any utility-based investor, if the pairs satisfying all these VaR constraints can be found, the elicited utility exists. Conversely, if no pair exists, the investor is not utility-based and thus also falls outside the scope of \cite*{BS2001}. Consequently, by adopting the preference-fitting method, we need only verify the feasibility of the pair to readily ascertain the existence of the elicited utility under multiple VaR constraints.
\begin{figure}
    \centering
    \includegraphics[width=0.6\linewidth]{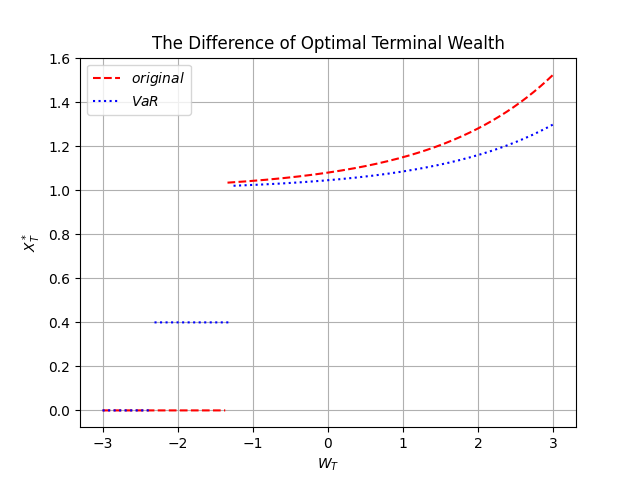}
    \caption{The difference of optimal terminal wealths. The values of the parameters are $m=1, r=0.05,T=1,\theta=0.25,\sigma=0.2,k_1=1,k_2=2.25,\gamma_1=\gamma_2=0.6,B_1=1,L=0.4,x_0=0.9,\alpha=0.01$.}
    \label{difference}
\end{figure}

Figure \ref{difference} illustrates a comparison of the terminal returns for the S-shaped utility with and without VaR constraints. When $W_T\in(-3,-2.319)$, both strategies result in a terminal wealth 0. In the range $(-2.319,-1.347)$, the VaR constraint yields an optimal wealth that is 0.4 higher than that of the original portfolio.  For $W_T \in (-1.267, 3)$, the wealth under the VaR constraint is slightly lower than the original wealth. Although the difference is increasing with respect to $W_T$ in the interval $(-1.267, \infty)$, the probability of the event $\{ W_T > 3 \}$ is only approximately $0.001$, making it negligible (extreme market conditions introduced in Section \ref{sec:fittingmethod}). In summary, when market conditions are unfavorable, the optimal portfolio under the VaR constraint provides the investor with stable returns and safeguards against bankruptcy, albeit at the expense of a slight loss in more favorable market conditions.

Based on the above analyses, the martingale-duality method and the bijection in Section \ref{sec:bijection}, it is easy to conclude that when we adjust the hyperbolic expression to take smaller values (denoted by $h_i^k$) on $[\xi^{p_i^k}, \xi^{p_{i+1}^k}]$ to construct the fitted terminal wealth $X_T^k$ in Section \ref{sec:fittingmethod}, the resulting effect corresponds to that of introducing VaR constraints in the traditional optimization problem.
This suggests that when a utility-based investor wants to reduce the risk she faces, she should choose a smaller anticipated return value when market conditions are favorable, so as to raise the budget upper bound at $p = 1 - \alpha$. The essence of this trade-off is reflected in the endogenous budget constraint $\E[\xi_T\X_T^*(\xi_T)]=x_0$ of the Black-Scholes model.

Overall, under the VaR constraints considered in the literature of \cite*{BS2001}, the advantage of the preference-fitting framework over the original optimization framework lies primarily in its intuitiveness. First, our method offers clearer practical significance: under the traditional framework, classical indicators such as the probability of bankruptcy are obtained from calculations using classical methods (e.g., the martingale-duality method), making it difficult for investors to predict them based solely on the form of the utility function; in contrast, the preference-fitting method provides investors with clear references (e.g., budget bounds) for assessing the reasonableness of their strategies. Second, our method possesses stronger operability: directly specifying the form of the utility function may lead to theoretical results for the optimal portfolio that contradict the investor's actual expectations. As illustrated in Figure \ref{difference}, the investor either receives wealth exceeding 1 or wealth below 0.4. This is a scenario we believe rarely occurs in practice and such a ``cliff-like'' strategy tends to have high fluctuations, as noted in Section 7 of \cite*{LLMV2024}. In contrast, the preference-fitting method allows the investor to easily achieve a more evenly distributed wealth profile, corresponding to a more stable strategy. Third, because the VaR constraints align closely with the definition of pairs, a utility-based investor can directly determine the existence of a portfolio by checking whether the pairs satisfy the VaR constraints, whereas the results of \cite*{BS2001} rely on the existence of Lagrange multipliers, which, when the number of VaR constraints is large, are determined by a system of coupled algebraic equations whose solvability has not been generally studied.

\section{Concluding Remarks}\label{sec:conclusion}
We define the concept of the elicited utility and set the criterion as expected utility maximization in the Black-Scholes model. 
If an investor's desired terminal wealth increases as market conditions improve and the investor seeks to maximize profits, she is qualified as utility-based and the elicited utility exists. Our elicitation method embeds the investor's market view that higher returns entail greater risk, which mathematically manifests as budget bounds.
The specific solution procedure is carried out through the bijection defined in Section \ref{sec:bijection}. Moreover, for operational convenience, the investor only needs to provide a finite number of intuitive probability-wealth pairs, from which we derive their fitted terminal wealth, fitted wealth process, and fitted portfolio. Through the PHARA approximation, we establish convergences in the sense of almost surely, $L^r$, and uniform convergence.
Finally, we compare the fitted portfolio with the optimal portfolio under VaR constraints, emphasizing the advantages of our method, including its intuitiveness, analytical tractability, implementation convenience, and the avoidance of discussing the existence of Lagrange multipliers. The elicited utility functions share the same ARA and RRA functions, which reflect the investor's preferences toward different risks. 
Through the elicited utility, we can compare and calibrate against the utility functions assumed in traditional models, and also use it to test the satisfaction level of different strategies for an investor.

We finally highlight the advantages of the probability-wealth pair from perspectives of both classic theory and modern technology. First, it is an intuitive concept that is easy for investors to understand. The pair is also introduced in the choice surveys in \cite*{kahneman1979d} to elicit a utility function. The bijection in Section \ref{sec:bijection} further indicates that the preferences contained in the pairs reflect characteristics of the utility function. In recent years, there has been increasing interest in inverse reinforcement learning (IRL), which involves inferring an agent's reward function based on its observed behavior; see, e.g., \cite*{SP2021} and \cite*{cheng2023elicitingriskaversioninverse}. Our idea of probability-wealth pairs is similar to that logic. Second, the pairs demonstrate efficient data utilization, requiring a limited number of samples for analysis. With the future help of AI agents, it may be feasible to use discrete data points of the preference-fitting method in robo-advising and FinTech. 



%
%
%

\section*{Acknowledgment}{Zongxia Liang acknowledges   financial support from the National Natural Science Foundation of China (Grant Nos. 12271290,12371477).
Yang Liu acknowledges financial support from the National Natural Science Foundation of China (Grant No. 12401624), The Chinese University of Hong Kong (Shenzhen) University Development Fund (Grant No. UDF01003336) and Shenzhen Science and Technology Program (Grant No. RCBS20231211090814028, 2025TC0010) and is partly supported by the Guangdong Provincial Key Laboratory of Mathematical Foundations for Artificial Intelligence (Grant No. 2023B1212010001). 
The authors are grateful to members of the group of Financial Mathematics and Risk Management at The Chinese University of Hong Kong (Shenzhen), and  the members of the group of Actuarial Sciences and Mathematical Finance at the Department of Mathematical Sciences, Tsinghua University for their feedback and useful conversations.}


\bibliographystyle{plainnat}
\bibliography{bio2}

\appendix
\section{
Proofs of Some Auxiliary Lemmas}\label{appendix_A}
\subsection{Proof of Lemma \ref{inverse}}
We first prove that, for any function $V\in \mathcal{A}$, there exists a unique function $f_r\in \mathcal{B}$ satisfying $V(x)=\int_1^{x} f_r(t) \d t$, $ \forall x \in [0,\infty)$, and $f_r(x) = V'_+(x)$ a.e..

In fact, as $V$ is concave,  the functions $V'_+(x)$ and $V'_-(x)$ are well defined for any $x\in(0,\infty)$, and $V'_+(x)=V'_-(x)$ a.e.. Moreover, $V'_+(x)$ is decreasing and right continuous, and thus $V'_+$ is integrable on any interval $[a,b] \subset (0,\infty)$. Based on  the definitions of $\mathcal{A}$ and $\mathcal{B}$, we easily conclude that $V'_+\in\mathcal{B}$.

We have $V(x) = \int_1^xV'_+(t) \d t$ holds for any $x> 0$ due to the absolute continuity of $V$. Letting $x\to 0^+$, we have $V(0) = \int_1^0 V'_+(t) \d t$ which is allowed to take the value of $-\infty$.

Then we prove the uniqueness. Suppose that there is another right continuous function $g_r \in \mathcal{B}$ satisfying $V(x)=\int_1^xg_r(t) \d t$ for any $x\geq 1$, and there exists $x_0 \geq 1$ such that $g_r(x_0)>V'_+(x_0)$ or $g_r(x_0)<V'_+(x_0)$. Then, for some $\delta>0$, we have $\int_x^{x+\delta}f_r(t) \d t > \int_x^{x+\delta}V'_+(t) \d t$ or $\int_x^{x+\delta}f_r(t) \d t < \int_x^{x+\delta}V'_+(t) \d t$ due to the fact that $V'_+$ and $g_r$ are right continuous, which leads to a contradiction. Thus, we obtain the uniqueness. Similarly, we conclude that the uniqueness holds on  $ (0,1]$. Then, letting $x \to 0^+$, we obtain the  uniqueness  for $x=0$.

We prove that $D_+$ is a bijection. On the one hand, using the result above, for any $V \in \mathcal{A}$, we have $D_+(V)\in \mathcal{B}$ and $T_0 \circ D_+(V)=V$. On the other hand, for any $f_r \in \mathcal{B}$, if we denote $V_r(x)= \int_1^x f_r(t) \d t$, then we easily get that $V_r(1)=0$ and $V_r$ is increasing. The continuity on $(0,x)$ for any $x>0$ is followed by the property of the definite integral. If $\int_0^1f_r(x)\d x<\infty$,  letting $x \to 0^+$, we get the continuity of $V$ at $0$. Otherwise, it happens that $V_r(0)=-\infty$, and the continuity follows as we have extended the definition of continuity in Remark \ref{class}. Next we prove the concavity of $V_r$, and thus $V_r \in \mathcal{A}$ holds. Indeed, for any $\lambda \in (0,1)$, and $x_1,x_2\in (0,\infty)$ with $x_1<x_2$, we prove$$
\lambda \int_1^{x_1}f_r(x)\d x +(1-\lambda) \int_1^{x_2}f_r(x)\d x \leq\int_1^{\lambda x_1+(1-\lambda)x_2}f_r(x)\d x.
$$
Rearranging the above inequalities and using the substitution method in integrals, it suffices to show  $$
\int_0^{x_2-x_1}f_r(x_1+t)\d t \leq \int_0^{x_2 -x_1}f_r(x_1+(1-\lambda )t)\d t,
$$
and the result follows from the fact that $f_r$ is decreasing. Then  $V_r$ is concave as it is continuous at $0$. It is easy to see that the inequality holds if $x_1=0$. Hence  $V_r\in\mathcal{A}$ and $$
\lim_{\Delta x\to0^+}\frac{\int_{x}^{x+\Delta x}f_r(t)\d t}{\Delta x}=f_r(x).
$$
Hence, $(V_r)'_+(x)=f_r(x)$ due to the right continuity of $f$. Therefore,  $D_+ \circ T_1(f_r) =f_r$. Thus, $D_+$ is a bijection and $D_+^{-1} = T_0$.

The definitions of $T_1$ and $T_2$ are similar to the generalized inverse in analysis and we can  directly verify that $T_1$ and $T_2$ are mappings from $\mathcal{B}$ to $\mathcal{C}$ and $\mathcal{C}$ to $\mathcal{B}$, respectively. To prove that they are bijections, we firstly prove  $T_2\circ T_1(f_r(x))=f_r(x)$ for any $f_r\in \mathcal{B}$ and $x \in[0,\infty)$, and it is similar to verify that $T_1\circ T_2(f_l)=f_l$. We need to prove that the following equality holds:
$$
\sup_{k\geq 0}\left\{\sup_{y \geq 0}\{f_r(y)\geq k\}>x \right\}=f_r(x).
$$
Assume that for some $x\geq0$, we have $\sup\limits_{k\geq 0}\left\{\sup\limits_{y\geq 0}\{f_r(y)\geq k\}>x\right\}>f_r(x)$. Then there exists an $\epsilon_1 >0$ satisfying $\sup\limits_{y\geq 0}\{ f_r(y)\geq f_r(x)+\epsilon_1\}>x$. Thus, there exists an $\epsilon_2 >0$ satisfying $f_r(x+\epsilon_2)\geq f_r(x)+\epsilon_1$. This contradicts the fact that $f$ is decreasing. Assume that for some $x\geq0$, we have $\sup\limits_{k\geq 0}\left\{\sup\limits_{y\geq 0}\{f_r(y)\geq k\}>x\right\}<f_r(x)$. There exists an $\epsilon_3 >0 $ satisfying $\sup\limits_{y\geq 0}\{f_r(y)\geq f_r(x)-\epsilon_3\}\leq x$. Then for any $\epsilon_4>0$, we have $f_r(x+\epsilon_4) < f_r(x)-\epsilon_3$, which contradicts the right continuity of $f_r$. Thus $T_2\circ T_1(f_r(x))=f_r(x)$.
\subsection{Proof of Lemma \ref{nu_covergence}}\label{l4.1}
In our proof, if $n$ is sufficiently large, $V_n$ will only take the form of a power base in a fixed closed interval that does not have 0 as an endpoint, and we will no longer make a special note of this.

(1) The function $V'_+$ is integrable on any interval $[X_1,X_2]\subset (0,\infty)$ and satisfies $V(x)=\int_1^x V'_+(t)\d t$ for any $x>0$. Without loss of generality, we let $x>1$. Then  $$|V(x)-V_n(x)|=\left|\int_1^{x}V_+'(t)-(V_n)'_+(t)\d t \right|\leq \int_1^{X_2}\left|V_+'(t)-(V_n)'_+(t)\right|\d t.$$ The construction of $V_n$ implies that the value of $(V_n)'_+$ on $[1,x]$ is bounded between the lower and upper Darboux sums of $V'$ with respect to the $n$-th partition on the shortest interval $[s_i^n,t_j^n]$ that covers $[1,x]$. Therefore, based on the Riemann integrability of $V'$, we conclude that $V_n$ converges uniformly to $V$ on the interval $ [1,x]$.

(2) We first prove that for any $n\in \N^*$ and $y>V'_+(n)$, we have $|f(y)-f_n(y)|\leq d_n$. There exists $1\leq i\leq M(n)$ such that $s_i^n \leq f(y)< t_i^n$. Therefore, based on the definition of $f$ and the proof of Theorem \ref{th:TL}, we have $D_+V(s_i^n)\geq y$ and $D_+V(t_i^n)< y$. Then we have $f_n(y)\in[s_i^n,t_i^n)$ due to  $(V_n)'_+(s_i^n)=V'_+(s_i^n)$ and $(V_n)'_+(t_i^n)=V'_+(t_i^n)$. Thus $|f(y)-f_n(y)|<t_i^n-s_i^n\leq d_n$. 

For any $\epsilon>0$ and $\nu>0$, using DCT, there exists $N_1 \in \N^*$ such that, for any $n>N_1$, $\E\left[(\xi_Tf(\nu\xi_T))^2\id_{\left\{f(\nu\xi_T)\in \left[t^n_{M(n)},\infty \right) \right\}}\right]<\epsilon.$ Moreover,
$$\E\left[\left[(\xi_Tf_n(\nu\xi_T)-\xi_Tf(\nu\xi_T))^2\right]\id_{\left\{f(\nu\xi_T)\in \left[0,t^n_{M(n)} \right) \right\}}\right] \leq d_n^2\cdot \E\left[\xi_T^2\id_{\left\{f(\nu\xi_T)\in \left[0,T(n) \right) \right\}}\right].
$$ 
In addition,  if $x\geq T(n)$, we have $f_n(\nu\xi_T)=T(n)+\frac{1}{\nu\xi_T}-\frac{1}{V'_+(T(n))}$, and if $f(\nu\xi_T)\in[T(n),\infty)$, we have $\nu\xi_T\leq V'_+(T(n))$. Therefore, using  Cauchy-Schwarz's inequality, we obtain
\begin{align*}\small
&\E\left[(\xi_Tf_n(\nu\xi_T))^2\id_{\{f(\nu\xi_T)\in[t_{M(n)}^n,\infty)\}}\right] 
\leq 3\E\left[\xi_T^2\left( f^2(\nu\xi_T)+\frac{2}{(\nu\xi_T)^2}\right)\id_{\{f(\nu\xi_T)\in[T(n),\infty)\}}\right].
\end{align*}
As such, there exists $N_2\in \N^*$ such that, for $n>\max{\{N_1,N_2\}}$, we have $$\E\left[(\xi_Tf_n(\nu\xi_T))^2\id_{\{f(\nu\xi_T)\in[t_{M(n)}^n,\infty)\}}\right] <\epsilon.$$
Hence, for any $n>\max\{N_1,N_2\}$,
\begin{align*}
\E\left[\left(\xi_Tf_n(\nu\xi_T)-\xi_Tf(\nu\xi_T)\right)^2\right]& = \E\left[\left|\left(\xi_Tf_n(\nu\xi_T)-\xi_Tf(\nu\xi_T)\right)^2\right|\id_{\{f(\nu\xi_T)\in(0,T(n)]\}}\right] \\ 
&\quad  +\E\left[\left|\left(\xi_Tf_n(\nu\xi_T)-\xi_Tf(\nu\xi_T)\right)^2\right|\id_{\{f(\nu\xi_T)\in(T(n),\infty)\}}\right] \\
&\leq d_n^2\cdot \E\left[\xi_T^2\id_{\left\{f(\nu\xi_T)\in \left(0,T(n) \right] \right\}}\right]+4\epsilon.
\end{align*}
Because $d_n \to 0$ as $n \to \infty$ and $\E\left[\xi_T^2\right]<\infty$, there exists $N_3\in \N^*$ such that,  for  any  $n>\max\{N_1,N_2,N_3\}$, we have $$\E\left[\left(\xi_Tf_n(\nu\xi_T)-\xi_Tf_n(\nu\xi_T)\right)^2\right]<5\epsilon,$$ by which  for any $\nu>0$, we have $\lim\limits_{n\to \infty}\E\left[(\xi_Tf_n(\nu\xi_T)-\xi_Tf(\nu\xi_T))^2\right] \to 0$. Thus $\left\{\E\left[(\xi_Tf_n(\nu\xi_T))^2\right]\right\}_{n\geq 1}$ has a uniform upper bound.

(3) First, we prove that $g$ and $g_n$ are strictly decreasing. Based on the fact that $f$ and $f_n$ are decreasing, we have that $g$ and $g_n$ are decreasing. We have  $x \mapsto f(\nu_1x)-f(\nu_2x)$ for any $\nu_1<\nu_2$ is nonnegative, left-continuous and not identically equal to $0$, which means that there exist an interval $[x_1,x_2]$ and $\epsilon>0$ such that $f(\nu_1x)-f(\nu_2x)>\epsilon$ holds on $[x_1,x_2]$. Therefore, we obtain $g(\nu_1)>g(\nu_2)$ for any $\nu_1<\nu_2$, and thus $g$  strictly decreases. Similarly, we have that $g_n$ strictly decreases for any $n\in \N^*$. 

Then we prove the continuity. Noting that when $\nu^*$ is fixed, for any $\nu$ satisfying $\nu^*/2<\nu<2\nu^*$, we have $\E\left[(\xi_Tf(\nu^* \xi_T/2))^2\right]>\E\left[(\xi_Tf(\nu \xi_T))^2\right]$ and hence the events where $\lim\limits_{\nu \to \nu^*}f(\nu\xi_T)=f(\nu^*\xi_T)$ fail to form a set of zero probability measure as the discontinuity points of $f$ constitute a at most countable set. Using the property of uniformly integrable random variables, we obtain $\lim\limits_{\nu \to \nu^*}\E[\xi_Tf(\nu \xi_T)]=\E[\xi_Tf(\nu^* \xi_T)]$, which indicates the continuity of $g$. The continuity of $g_n$ follows similarly. 

Last, we prove the existence, uniqueness and convergence of $\{\nu^*_n\}_{n\geq 1}$. The values of the limits $\lim\limits_{\nu \to \infty}g(\nu)=\lim\limits_{\nu \to \infty}g_n(\nu)=0$ and $\lim\limits_{\nu \to 0^+}g(\nu)=\lim\limits_{\nu \to 0^+}g_n(\nu)=\infty$ are derived directly from DCT and Lévy's monotone convergence theorem. Thus, the existence and uniqueness of $\nu^*$ and $\nu^*_n$ ($n \in \N^*$) hold using the fact that $g$ and $g_n$ are strictly decreasing and continuous. There exists $N\in \N^*$ such that $|f_n(\nu^* \xi_T) - f(\nu^* \xi_T)| < \epsilon$ for any $n > N$. Therefore, we have $f_n(\nu^* \xi_T) \to f(\nu^* \xi_T)$, $n \to \infty$, given $\nu^*\xi_T$ is fixed. Due to the fact that $$\lim_{n \to \infty}\xi_Tf_n(\nu^* \xi_T)=\xi_Tf(\nu^* \xi_T) \quad \text{a.s.},$$ and $\{\E\left[(\xi_Tf_n(\nu^* \xi_T))^2\right]\}_{n\geq 1}$ has a uniform upper bound according to (2), we obtain $$\lim_{n\to\infty}\E[\xi_Tf_n(\nu^*\xi_T)]=\E[\xi_Tf(\nu^*\xi_T)],$$ i.e., $\lim\limits_{n\to \infty}g_n(\nu) = g(\nu)$. For any $0<\epsilon<\nu^*$, where $\nu^*$ is the unique solution of $g(\nu) =x_0$, because $g$ stricly decreases, we have $g(\nu^*-\epsilon)>0$ and $g(\nu^*+\epsilon)<0$. Define $\delta = \min\{|g(\nu^*-\epsilon)|,|g(\nu^*+\epsilon)|\}$. Then there exists $N \in \mathbb{N}^*$ such that for any $n>N$, we have $g_n(\nu^*-\epsilon)>0$ and $g_n(\nu^*+\epsilon)<0$, which indicates that $\nu_n^*\in(\nu^*-\epsilon,\nu^*+\epsilon)$ holds for $n>N$. Thus, $\nu^*_n\to \nu^*$ follows.

\subsection{Proof of Lemma \ref{lemma:L^2-co}}
Because the sequence $\{\nu^*_n\}$ is bounded, we have $\p \{f(\nu_n^*\xi_T)\in(0,T(n))\}\to 1$. Imitating the proof of Lemma \ref{nu_covergence}, we obtain that there exists $N_1\in \N^*$ such that for any $n>N_1$, 
$\E\left[\left(\xi_Tf_n(\nu_n^*\xi_T)-\xi_Tf(\nu_n^*\xi_T)\right)^2\right]
\leq d_n^2\cdot \E\left[\xi_T^2\id_{\left\{f(\nu_n^*\xi_T)\in \left(0,T(n) \right) \right\}}\right]+4\epsilon\leq 5\epsilon,$ 
i.e.,   $\lim\limits_{n\to \infty}\E\left[\left(\xi_Tf_n(\nu_n^*\xi_T)-\xi_Tf(\nu_n^*\xi_T)\right)^2\right] = 0$.

Moreover, if we  define $\nu_{\text{inf}}=\inf\limits_{n\geq 1}\{\nu^*_n\}>0$, then, using DCT, there exists $0<x_1<x_2$ such that 
$\E\left[(\xi_Tf(\nu_{inf}\xi_T))^2\id_{\left\{f(\nu^*\xi_T)\in \left(0,x_1 \right)\cup (x_2,\infty) \right\}}\right]<\epsilon.$ 
Therefore, $\E\left[(\xi_Tf(\nu_{n}^*\xi_T))^2\id_{\left\{f(\nu^*\xi_T)\in \left(0,x_1 \right)\cup (x_2,\infty) \right\}}\right]<\epsilon$ holds for any $n\in\N^*$ due to the monotonicity of $f$. Using $\xi_Tf(\nu^*_n\xi_T) \to \xi_Tf(\nu^*\xi_T)$, a.s., we have that $\{\xi_Tf(\nu^*_n\xi_T)\}_{n\geq 1}$ has a uniform upper bound $\frac{V'_-(x_1)}{\nu_{\text{inf}}}T_L\left(\frac{V'_+(x_2)\nu_{\text{inf}}}{\nu^*}\right)$ for any $n\in \N^*$ and $\xi_T$ satisfying $f(\nu^*\xi_T)\in \left[x_1 ,x_2\right] $. 

Using DCT yields
$$\lim_{n\to\infty}\E\left[(\xi_Tf(\nu_{n}^*\xi_T)-\xi_Tf(\nu^*\xi_T))^2\id_{\left\{f(\nu^*\xi_T)\in \left[x_1,x_2 \right] \right\}}\right] = 0. $$
Thus, $\lim\limits_{n\to\infty}\E\left[(\xi_Tf_n(\nu_n^*\xi_T)-\xi_Tf(\nu^*\xi_T))^2\right]=0$.
\begin{remark}
If the domain of $V$ is $(-\infty,\infty)$ as stated in Remark \ref{remark:domain}, we also have 
$$\E\left[\left(\frac{-T_1(n)V'_+(-T_1(n))}{\nu}\right)^2\id_{\{f(\nu\xi_T)<-T_1(n)\}}\right]
\leq\E\left[(\xi_Tf(\nu \xi_T))^2\id_{\{f(\nu\xi_T)<-T_1(n)\}}\right]$$ 
using the fact that $\nu\xi_T \geq V'_+(-n)$ on  $\{\omega : f(\nu\xi_T)<-n\}$. Thus (2) still holds in this case. Additionally, when the first partition $[0,t_i^n)$ uses a logarithmic base as stated in Remark \ref{remark:domain}, we can verify that (2)  holds in a slightly different manner.
\end{remark}
\subsection{Proof of Lemma \ref{prop:1}}
In fact, the sequence $\left\{\E\left[\left(\xi_T\right)^{r_1}\left|\left(\ln \left(\xi_T\right)\right)\right|^{r_2}f_n^{r_3}(\nu^*_n\xi_T)\right]\right\}_{n \in \N^*}$ has a uniform upper bound. This holds based on  Hölder's inequality and one of the upper bound is 
\begin{equation}\label{eq:uni_bound}
\(\E\left[\left(\xi_Tf(\nu_{inf}\xi_T)\right)^2\right]\)^{\frac{r_3}{2}} \(\E\left[\xi_T^{\frac{4(r_1-r_3)}{2-r_3}}\right]\)^{\frac{2-r_3}{4}} \(\E\left[\left|\ln \xi_T\right|^{\frac{4r_2}{2-r_3}}\right]\)^{\frac{2-r_3}{4}}. 
\end{equation}
Using the relationship between $F_t$ and $F_0$, 
\begin{align*}\small
\int_{0}^{\infty}x^{r_1}|\ln x|^{r_2}f_n^{r_3}(\nu^*_n\xi_tx)\d F_{t}(x)
&\leq \sqrt{\frac{T}{T-t}}\int_{0}^{\infty}|\ln x|^{r_2}\left(f_n(\nu^*_nx)\right)^{r_3}C_1(\xi_t)x^{C_2(\xi_t)}\d F_{0}(x)\\ & \quad +\sqrt{\frac{T}{T-t}}\int_{0}^{\infty}\left(f_n(\nu^*_nx)\right)^{r_3}C_3(\xi_t)x^{C_2(\xi_t)}\d F_{0}(x),
\end{align*}
where 
$$
\begin{aligned}
C_1(\xi_t)&=\xi_t^{-r_1}\exp\left\{\frac{1}{2}\left[\frac{(T-t)\left(r+\frac{||\boldsymbol{\theta}||^2}{2}\right)^2t-(\ln(\xi_t(\omega)))^2+2\ln (\xi_t(\omega))(T-t)\left(r+\frac{||\boldsymbol{\theta}||^2}{2}\right)}{||\boldsymbol{\theta}||^2(T-t)}\right]\right\},\\
C_2(\xi_t)&=\frac{\ln(\xi_t(\omega))}{||\boldsymbol{\theta}||^2(T-t)}+r_1,\  \ C_3(\xi_t)=|\ln \xi_t|^{r_2}\cdot C_1(\xi_t)
\end{aligned}
$$
are real-valued functions evaluated at $\xi_t$. Then 
\begin{equation}\label{eq:upb}
\begin{aligned}
\left|\int_{0}^{\infty}x^{r_1}|\ln x|^{r_2}f_n^{r_3}(\nu^*_n\xi_tx)\d F_{t}(x)\right|&\leq C_1(\xi_t)\sqrt{\frac{T}{T-t}}\left|\E\left[\xi_T^{C_2(\xi_t)}|\ln (\xi_T)|^{r_2}(f_n(\nu_n^*\xi_T))^{r_1}\right]\right|\\&\quad + C_3(\xi_t)\sqrt{\frac{T}{T-t}}\left|\E\left[\xi_T^{C_2(\xi_t)}(f_n(\nu_n^*\xi_T))^{r_1}\right]\right|
<\infty.
\end{aligned}
\end{equation}
Therefore, according to Eqs. (\ref{eq:uni_bound}) and (\ref{eq:upb}), the sequence $\left\{\left|\int_{0}^{\infty}x^{r_1}|\ln x|^{r_2}f_n^{r_3}(\nu^*_n\xi x)\d F_{t}(x)\right|\right\}_{n\in \N^*}$ has a uniform upper bound for fixed $\xi>0$ and any $r_1\in \R,r_2 \in [0,\infty),r_3\in(1,2),t\in[0,T)$. 
\subsection{Proof of Lemma \ref{lemma:twice-dif}}
\noindent
(1) We first prove the differentiability with respect to $t$. Based on the proof martingale-duality method, we have 
\begin{equation}\label{eeq1}
\mathcal{X}_t^*(t,\xi)=\int_0^\infty f(\nu^*\xi x)K(x,t)\d x .
\end{equation}
By Lagrange's mean value theorem, for fixed $t>0$ and $\xi>0$, and any $h\in(-\frac{t}{2},\frac{t}{2})$, there exist three constants only depending on $t$, which are denoted by $\tilde{C}_1(t)$, $\tilde{C}_2(t)$, and $\tilde{C}_3(t)$ such that $$\left|\frac{K(x,t+h)-K(x,t)}{h}\right| \leq \left(\tilde{C}_1(t)(\ln x)^2+\tilde{C}_2(t)|\ln x|+\tilde{C}_3(t)\right)K(x).$$
Therefore, Lemma \ref{prop:1} ensures the premise of the DCT. Hence, we have 
$$
\frac{\partial}{\partial t}\int_0^\infty f(\nu^*\xi x)K(x,t)\d x=\int_0^\infty f(\nu^*\xi x) \frac{\partial}{\partial t}K(x,t)\d x,$$ 
i.e., $\mathcal{X}_t^*$ is differentiable with respect to $t$.  Similarly, we have the higher-order differentiability of $\mathcal{X}_t^*$ with respect to $t$.

Then we prove the differentiability with respect to $\xi$. Here, with a little abuse of notation, we denote $K(x,t)=K(x)$, because we only focus on the variable $x$. For fixed $\xi>0$ and $h\in \left(-\frac{\xi}{2},\frac{\xi}{2}\right)$, we have $$\frac1h \left(\int_0^\infty f(\nu^*(\xi+h)x)K(x)\d x-\int_0^\infty f(\nu^*\xi x)K(x)\d x\right)=I_1(h)+I_2(h),$$
where $$
I_1(h)=\frac{1}{h}\int_0^\infty f(x)\left(K\left(\frac{x}{\nu^*(\xi+h)}\right)-K\left(\frac{x}{\nu^*\xi}\right)\right)\frac{1}{\nu^*\xi}\d x$$ and $$I_2(h)=\frac{1}{h}\int_0^\infty f(x)K\left(\frac{x}{\nu^*(\xi+h)}\right)\left(\frac{1}{\nu^*(\xi+h)}-\frac{1}{\nu^*\xi}\right)\d x.$$
As $K'(0)=0$ and $K'$ is decreasing on the interval $[0,\delta)$, where $\delta>0$ is a constant, using Lagrange's mean value theorem implies
$$\Big|\frac{1}{\nu^*\xi h}f(x)\left(K\left(\frac{x}{\nu^*(\xi+h)}\right)-K\left(\frac{x}{\nu^*\xi}\right)\right)\Big|\leq \Big|\frac{xf(x)}{\nu^2\xi^2(\xi+h)}K'\Big(\frac{x}{\nu^*(\xi+h)}\Big)\Big|$$
for any $0\leq x <2\nu^*\xi\delta$ and $-\frac\xi2<h<\frac{h}{2}$. Then, for any $\epsilon>0$, Lemmas \ref{prop:1}-\ref{prop:2} indicate that there exists $\delta_1\in (0,\delta)$ such that $$\frac{1}{h}\int_0^{\delta_1}\Big| f(x)\left(K\left(\frac{x}{\nu^*(\xi+h)}\right)-K\left(\frac{x}{\nu^*\xi}\right)\right)\frac{1}{\nu^*\xi}\Big|\d x<\epsilon.$$
Similarly, for some $\delta_2>0$, we have
$$\frac{1}{h}\int_{\{x\in(0,\delta_1)\cup(\delta_2,\infty)\}}\Big| f(x)\left(K\left(\frac{x}{\nu^*(\xi+h)}\right)-K\left(\frac{x}{\nu^*\xi}\right)\right)\frac{1}{\nu^*\xi}\Big|\d x<\epsilon.$$
Moreover, by Lemma \ref{prop:1}, we have $\left\{\frac{1}{\nu^*\xi h}f(x)\left(K\left(\frac{x}{\nu^*(\xi+h)}\right)-K\left(\frac{x}{\nu^*\xi}\right)\right)\right\}_{h \in \left(-\frac{\xi}{2},\frac{\xi}{2}\right)}$ is uniform integrable on $[\delta_1,\delta_2]$. As such, 
$
\lim_{h\to 0}I_1(h)=-\frac{1}{\xi}\int_0^\infty xf(\nu^*\xi x)K'(x)\d x .$ Similarly, $\lim_{h\to0}I_2(h)=-\frac1\xi \int_0^\infty f(\nu^*\xi x)K(x)\d x.$  
Thus, $$\frac{\partial}{\partial \xi}\mathcal{X}^*(t,\xi)=-\frac{1}{\xi} \int_0^\infty f(\nu^*\xi x)(xK(x))'\d x.$$
Using the same way, we can prove the higher-order differentiability with respect to $\xi$.
\\
(2) As (1) holds, we can directly use Itô's formula and compare the  coefficient in front of the stochastic term to obtain  $$\d X_t^* = \dots \d t + \boldsymbol{\pi}_t^\top \boldsymbol{\sigma} \d \mathbf{W}_t=\dots \d t +\xi_t \frac{\partial \mathcal{X}_t^*(t,\xi_t)}{\partial \xi_t}\boldsymbol{\theta}\d \mathbf{W}_t,
$$
and thus the result follows.

\end{document}